\documentclass[a4paper,english,numberwithinsect,final]{lipics-v2016}

\usepackage[utf8]{inputenc}
\usepackage[T1]{fontenc}

\usepackage[notcite,notref]{showkeys}

\usepackage{color}

\usepackage[layout=footnote]{fixme}
\FXRegisterAuthor{sm}{asm}{SM}

\usepackage[all]{xy}
\SelectTips{cm}{}

\newdir{ >}{{}*!/-10pt/\dir{>}}
\newdir{ (}{{}*!/-10pt/\dir^{(}}

\usepackage{preamble_full}
\usepackage{dsfont}
\usepackage{enumerate}
\usepackage{xspace}

\usepackage{microtype}

\title{On Corecursive Algebras for Functors Preserving Coproducts\footnote{Full version of the paper published in: Proc.~7th Conference on Algebra and Coalgebra in Computer
Science (CALCO 2017), LIPIcs 72, 3:1--3:15.}}
\titlerunning{Corecursive Algebras}

\author[1]{Ji\v{r}\'i Ad\'amek}
\author[2]{Stefan Milius\footnote{Supported by Deutsche Forschungsgemeinschaft (DFG) under project MI~717/5-1}$^,$}
\affil[1]{Institut f\"{u}r Theoretische Informatik \\
  Technische Universit\"{a}t Braunschweig, Germany}
\affil[2]{Lehrstuhl f\"{u}r Theoretische Informatik\\
  Friedrich-Alexander-Universit\"{a}t Erlangen-N\"{u}rnberg, Germany}
\authorrunning{J.~Ad\'{a}mek and S.~Milius}
\Copyright{Ji\v{r}\'{i} Ad\'{a}mek and Stefan Milius}

\subjclass{F.3.2~Semantics of Programming Languages}

\keywords{terminal coalgebra, free algebra, corecursive algebra, hyper-extensive category}

\iffull\else
\EventEditors{Filippo Bonchi and Barbara K\"onig}
\EventNoEds{2}
\EventLongTitle{7th Conference on Algebra and Coalgebra in Computer
Science (CALCO 2017)}
\EventShortTitle{CALCO 2017 (full version)}
\EventAcronym{CALCO}
\EventYear{2017}
\EventDate{June 12--16, 2017}
\EventLocation{Ljubljana, Slovenia}
\EventLogo{}
\SeriesVolume{72}
\ArticleNo{3}
\fi 

\begin{document}
\maketitle
\begin{abstract}
  For an endofunctor $H$ on a hyper-extensive category preserving
  countable coproducts we describe the free corecursive algebra on $Y$
  as the coproduct of the terminal coalgebra for $H$ and the free
  $H$-algebra on $Y$. As a consequence, we derive that $H$ is a 
  \emph{cia functor}, i.e., its corecursive algebras are precisely the cias 
  (completely iterative algebras). Also all functors $H(-) + Y$ are
  then cia functors.  For finitary set functors we prove that,
  conversely, if $H$ is a cia functor, then it has the form
  $H = W \times (-) + Y$ for some sets $W$ and $Y$.
\end{abstract}

%
%
\section{Introduction}

Iteration and (co)recursion are of central importance in computer
science. A formalism for iteration was proposed by Elgot~\cite{elgot}
as iterative algebraic theories. Later Nelson~\cite{nelson} and
Tiuryn~\cite{tiuryn80} introduced iterative algebras for finitary
signatures which yield an easier approach to iterative theories. For
endofunctors $H$ there are two related notions of
algebras. \emph{Corecursive algebras} introduced by Capretta et
al.~\cite{cuv09} are those algebras $A$ such that every recursive
equation expressed as a coalgebra for $H$ has a unique solution
(i.e., a coalgebra-to-algebra morphism into $A$). The other notion,
\emph{completely iterative algebras} (or \emph{cia}, for short),
introduced by the second author \cite{m_cia}, are $H$-algebras $A$
with the stronger property that every recursive equation with
parameters in $A$ has a unique solution
(Definition~\ref{D:cia}). Corecursive algebras often fail to be
cias. In the present paper we study endofunctors such that every
corecursive algebra is a cia -- we call them \emph{cia functors}.

Our first result is that every endofunctor preserving countable
coproducts and having a terminal coalgebra is a cia functor
(Corollary~\ref{cor:ciacorec}). This is based on a description of the
free cia on an object $Y$ as a coproduct
\[
\nu H + FY
\] 
of the terminal coalgebra and the free algebra on $Y$
(Theorem~\ref{thm:coprod}). We deduce that, for $H$ preserving
countable coproducts and having a terminal coalgebra, we obtain cia
functors $H(-) + Y$ for all objects $Y$
(Corollary~\ref{C:strong}). All this holds in every \emph{hyper-extensive}
base category (Definition~\ref{dfn:hyper}), e.g., in sets, posets,
graphs and all presheaf categories.

In particular, if the base category is also cartesian closed, then
$X \mapsto W \times X + Y$ is a cia functor for every pair of objects $W$ and
$Y$. For finitary set functors we prove a surprising converse: the
only cia functors are those of the above form $X \mapsto W \times X + Y$.

\iffull
Finally, we investigate the Eilenberg-Moore algebras for the free cia
monad $T$. In general, these are characterized as the complete Elgot
algebras for $H$~\cite{amv_elgot}. In the setting of this paper the
monad $T$ is also the monad of free corecursive algebras. The
Eilenberg-Moore algebras for the latter monad were characterized as
Bloom algebras for accessible functors on locally presentable
categories~\cite[Theorem~4.15]{ahm14}. We prove that under our
assumptions on $H$ complete Elgot algebras and Bloom algebras for $H$
are the same (Theorem~\ref{thm:elgotbloom}). 
\else
Due to space constraints we omit some proof details. However, these
details may be found in the full version~\cite{am17} of our paper. 
\fi

\section{Preliminaries}


Throughout the paper $H$ denotes an endofunctor on a hyper-extensive
category (recalled below) having a terminal coalgebra
\[
\ter: \nu H \to H(\nu H).
\]
By the famous Lambek Lemma~\cite{lambek68}, the coalgebra structure
$t$ is invertible and its inverse makes $\nu H$ an $H$-algebra. 

We denote by
$
\Alg H
$
the category of $H$-algebras and their morphisms.

\begin{definition}[\cite{abmv08}]
  \label{dfn:hyper}
  A category is called \emph{hyper-extensive} if it has countable
  coproducts which are
  \begin{enumerate}[(1)]
  \item \emph{universal}, i.e., preserved by pullbacks along any
    morphism,
  \item \emph{disjoint}, i.e., coproduct injections are monomorphic and
    have pairwise intersection $0$ (the initial object), and
  \item \emph{coherent}, i.e., given pairwise disjoint morphisms $a_n:
    A_n \to A$, $n \in \N$, each of which is a coproduct injection,
    then their copairing $[a_n]_{n \in \N}: \coprod_{n \in \N} A_n \to
    A$ is also a coproduct injection.
  \end{enumerate}
\end{definition}

\begin{example}
  The categories of sets, posets, graphs, and presheaf categories are
  hyper-extensive.
\end{example}

\begin{rem}
  \begin{enumerate}[(1)]	
  \item We write $A+B$ for the coproduct of the objects $A$ and $B$
    and denote coproduct injections by $\inl: A \to A+B$ and
    $\inr: B \to A+B$.
  \item Recall that a category with finite coproducts is
    \emph{extensive} if it has pullbacks along coproduct injections
    and conditions~(1) and (2) are
    satisfied~\cite{clw93}. Equivalently, in a diagram of the following form
    \[
      \xymatrix{
        X
        \ar[d]_f \ar[r]^x
        &
        Z
        \ar[d]_h
        &
        Y 
        \ar[l]_-{y} 
        \ar[d]^g
        \\
        A \ar[r]_-{\inl}
        & 
        A + B
        &
        B \ar[l]^-{\inr}
      }
    \]
    the top row is a coproduct if and only if the squares are
    pullbacks. Another, more compact, equivalent characterization of
    extensivity states that the canonical functor
    $\C/A \times \C/B \to \C/(A+B)$ is an equivalence of categories
    for any pair of objects $A$ and $B$.
 
  \item The somewhat technical condition~(3) in
    Definition~\ref{dfn:hyper} is not a consequence of the other
    two. In fact, let $\C$ be the category of J\'onsson-Tarski algebras,
    i.e., binary algebras $A$ whose operation $A \times A \to A$ is a
    bijection. Then $\C$ has disjoint and universal countable (in fact, all)
    coproducts but is not hyperextensive~\cite{abmv08}.
  \end{enumerate}
\end{rem}
\begin{definition}[\cite{cuv09}]
  \label{dfn:corec}
  An algebra $a: HA \to A$ is called \emph{corecursive} if for every
  coalgebra $e: X \to HX$ there exists a unique algebra-to-coalgebra
  morphism $\sol e: X \to A$:
  \begin{equation}\label{eq:corec}
    \vcenter{\xymatrix{
      X \ar[r]^-{\sol e}
      \ar[d]_e 
      & 
      A
      \\
      HX
      \ar[r]_-{H\sol e}
      & 
      HA
      \ar[u]_a
    }}
  \end{equation}
\end{definition}
\begin{examples}\label{ex:corec}
  \begin{enumerate}[(1)]
  \item The terminal coalgebra $\nu H$ (considered as an algebra) is
    obviously corecursive. This is the initial corecursive
    algebra~\cite{cuv09}. 

    Furthermore, let $Y$ be an object of $\C$ and assume that the
    functor $H(-)+Y$ has a terminal coalgebra $TY$. Then its structure
    \[
      TY \xrightarrow{\alpha_Y} HTY + Y
    \]
    has an inverse which is the copairing of two morphisms denoted by 
    \[
      HTY \xrightarrow{\tau_Y} TY \xleftarrow{\eta_Y} Y.
    \]
    It follows that $TY$ is a coproduct of $HTY$ and $Y$ with the
    above coproduct injections. It is easy to show that $(TY, \tau_Y)$
    is a corecursive algebra. 
  \item The trivial terminal algebra $H1 \to 1$ is corecursive, and if
    $(A,a)$ is a corecursive algebra so is
    $(HA, Ha)$~\cite[Prop.~21]{cuv09}. Furthermore, if $\C$ has limits
    then corecursive algebras are closed under limits in the category
    of algebras for $H$~\cite[Prop.~2.4]{ahm14}. It follows that all
    members of the \emph{terminal-coalgebra} chain
    \[
      \xymatrix@1{1 & H1 \ar[l] & HH1 \ar[l] & \cdots \ar[l]}
    \]
    are corecursive algebras. 
  \item A particular instance of point~(1) is given by a signature
    $\Sigma = (\Sigma_n)_{n < \omega}$ of operation symbols with
    prescribed arity and considering the corresponding polynomial
    endofunctor $H_\Sigma$ on $\Set$ defined by
    \[
      H_\Sigma X = \coprod_{n < \omega} \Sigma_n \times X^n.
    \]
  \end{enumerate}
  For an operation symbol $\sigma \in \Sigma_n$ we write $\sigma(x_1,
  \ldots, x_n)$ in lieu of $(\sigma, (x_1, \ldots, x_n))$ for 
  elements in the summand of $H_\Sigma X$ corresponding to 
  $n < \omega$. The terminal coalgebra $\nu H_\Sigma$ is carried by the set of
  all \emph{$\Sigma$-trees}, i.e., rooted and ordered trees with nodes
  labeled in $\Sigma$ such that every node with $n$ children is
  labeled by an $n$-ary operation symbol. The algebraic operation of
  $\nu H_\Sigma$ is \emph{tree-tupling}: $t^{-1}$ assigns to
  $\sigma(t_1, \ldots, t_n)$ with $\sigma \in \Sigma_n$ and $t_i \in
  \nu H_\Sigma$, $i = 1 \ldots, n$, the $\Sigma$-tree obtained by
  joining the $\Sigma$-trees $t_1, \ldots, t_n$ by a root node labeled
  by $\sigma$. 

  For every set $Y$ we denote by
  \[
    T_\Sigma Y
  \]
  the algebra of all \emph{$\Sigma$-trees over $Y$}, i.e., $\Sigma$-trees
  whose leaves are labeled by constant symbols in $\Sigma_0$ or
  elements of $Y$. This is the terminal coalgebra for $H_\Sigma(-) + Y$,
  and therefore it is a corecursive algebra.
\end{examples}
\begin{rem}
  For a polynomial endofunctor $H_\Sigma$ on $\Set$ we can view a
  coalgebra $e: X \to H_\Sigma X$ as a system of \emph{recursive
    equations} over the set $X$ of (recursion) variables: for every
  variable $x \in X$ we have a formal equation
 \[
   x \approx \sigma(x_1, \ldots, x_n) = e(x).
 \]
 The map $\sol e$ in Definition~\ref{dfn:corec} is then a
 \emph{solution} of the system of equations in the $\Sigma$-algebra
 $A$: the commutative square~\eqref{eq:corec} states that $\sol e$
 turns the above formal equations into actual identities in $A$:
 \iffull\[\else\/$\fi
   \sol e(x) = \sigma^A(\sol e(x_1), \ldots, \sol e(x_n)).
 \iffull\]\else\/$\fi
\end{rem}
\begin{definition}[\cite{m_cia}]
  \label{D:cia}
  An algebra $a: HA \to A$ is called \emph{completely iterative} (or
  \emph{cia}, for short) if the algebra $[a, A]: HA + A \to A$ is
  corecursive for the endofunctor $H(-)+A$. That means that for
  every \emph{(flat) equation morphism} $e: X \to HX + A$ there exists
  a unique \emph{solution}, i.e., a unique morphism $\sol e$ such that
  square below commutes:
\begin{equation}\label{diag:sol}
  \vcenter{
    \xymatrix@C+2pc{
      X \ar[r]^-{\sol e} \ar[d]_e & A \\
      HX + A\ar[r]_-{H\sol e + A} & HA + A \ar[u]_{[a, A]}
    }
  }
\end{equation}
\end{definition}
\begin{examples}\label{ex:cia}
  \begin{enumerate}[(1)]
  \item If $H(-)+Y$ has a terminal coalgebra $TY$
    (cf.~Example~\ref{ex:corec}(1)), then $(TY, \tau_Y)$ is a cia. In
    fact, $(TY, \tau_Y)$ is a free cia on $Y$ with the universal
    morphism $\eta_Y$~\cite{m_cia}. 
  \item For a polynomial functor $H_\Sigma$ on $\Set$ the above
    example states that the algebra $T_\Sigma Y$ of all $\Sigma$-trees
    over $Y$ is the free cia on the set $Y$. Let us denote by
    \[
      C_\Sigma Y
    \]
    the subalgebra of $T_\Sigma Y$ given by all $\Sigma$-trees over
    $Y$ which have only a finite number of leaves labeled in $Y$ (and
    the remaining, possibly infinitely many, leaves are labeled in
    $\Sigma_0$). This algebra is corecursive but, whenever $\Sigma$
    contains an operation symbol of arity at least $2$, not a
    cia. Moreover, $C_\Sigma Y$ is the free corecursive algebra on
    $Y$~\cite{ahm14}.

    As a concrete example, consider the signature $\Sigma$ consisting
    of a single binary operation $\sigma$. Then the equation morphism
    $e: \{x_1, x_2\} \to H_\Sigma\{x_1,x_2\} + \{y\}$ given by the
    recursive equations
    \iffull
    \[
      x_1 \approx \sigma(x_1,x_2) \qquad\text{and}\qquad x_2 \approx y
    \]
    \else
    $x_1 \approx \sigma(x_1,x_2)$ and $x_2 \approx y$ \fi
    has the unique solution $\sol e: \{x_1,x_2\} \to T_\Sigma\{y\}$
    given as follows
    \[
      \sol e: x_1 \mapsto 
      \vcenter{
        \xy
        \POS (000,000) *+{\sigma} = "s0",
             (-05,-05) *+{\sigma} = "s1",
             (005,-05) *+{y} = "y1",
             (-10,-10) *+{\sigma} = "s2",
             (000,-10) *+{y} = "y2",
             (-15,-15) *+{\vdots} = "vd",
             (-05,-15) *+{y} = "y3",
        \ar@{-} "s0";"s1",
        \ar@{-} "s1";"s2",
        \ar@{-} "s2";"vd",
        \ar@{-} "s0";"y1",
        \ar@{-} "s1";"y2",
        \ar@{-} "s2";"y3"
        \endxy
      }
      \qquad
      x_2 \mapsto y.
    \]
    This demonstrates that $C_\Sigma\{y\}$ is not a cia because the above
    infinite $\Sigma$-tree is not contained in it. 
  \end{enumerate}
\end{examples}
\begin{definition} 
  A \emph{cia functor} is an endofunctor such that every corecursive
  algebra for it is a cia. (It the follows that cias and corecursive
  algebras coincide).
\end{definition}
\begin{notation}\label{not:FC}
  \begin{enumerate}[(1)]
  \item If a free $H$-algebra on $Y$ exists, we denote it by
    $FY$ and its structure and universal morphism by
    \[
      \phi_Y: HFY \to FY \qquad\text{and}\qquad
      \eta^F_Y: Y \to FY,
    \]
    respectively. 

    In the case of a polynomial set functor $H_\Sigma$, the free
    $\Sigma$-algebra $F_\Sigma Y$ is the subalgebra of $T_\Sigma Y$ on
    all finite $\Sigma$-trees over $Y$. 
  \item If a free corecursive $H$-algebra on $Y$ exists, we denote it
    by $CY$ and its structure and universal morphism by
    \[
      \psi_Y: HCY \to CY 
      \qquad{\text{and}}\qquad
      \eta^C_Y: Y \to CY,
    \]
    respectively.
  \end{enumerate}
\end{notation}

\section{Functors Preserving Countable Coproducts}
\label{sec:coprod}

\begin{assumption}\label{ass:3}
  In this and the subsequent section we assume that $H$ is an
  endofunctor on a hyper-extensive category having a terminal
  coalgebra and preserving countable coproducts.
\end{assumption} 

\begin{fact}[\cite{ArbibManes74}]
  A free algebra on $Y$ is 
  \[
    FY = H^*Y = \coprod_{n< \omega} H^nY
    \qquad
    \text{with coproduct injections $j_n: H^nY \to H^*Y$}.
  \]
  Its algebra structure and universal morphism are given by
  \[
    \phi_Y \cdot Hj_n = j_{n+1} \quad (n >0) \qquad\text{and}\qquad
    \eta^F_Y = j_0: Y \to H^* Y
  \]
  using that $HFY = \coprod_{n < \omega} H^{n+1} Y$.
\end{fact}
\begin{notation}
  We denote by 
  \[
    \sigma_Y: H^* Y = \coprod_{n< \omega} H^n Y \to Y +
    H\left(\coprod_{n<\omega} H^n Y\right) = Y + HH^*Y 
  \]
  the isomorphism inverse to $[\eta^F_Y,
  \phi_Y]: Y + HH^*Y \to
  H^*Y$. It is defined by the following commutative diagrams:
    \begin{equation}\label{diag:tri}
    \vcenter{
    \xymatrix{
      Y 
      \ar[rd]^-{\inl}
      \ar[d]_{j_0}
      \\
      H^*Y
      \ar[r]_-{\sigma_Y}
      &
      Y + HH^*Y
    }
    }
    \qquad\qquad
    \vcenter{
    \xymatrix{
      H^n Y
      \ar[r]^-{Hj_{n-1}}
      \ar[d]_{j_n}
      &
      HH^*Y
      \ar[d]^{\inr}
      \\
      H^*Y 
      \ar[r]_-{\sigma_Y}
      &
      Y + HH^* Y
    }}\qquad\text{for $n > 0$}.
  \end{equation}
\end{notation}
\iffull
\begin{lemma}\label{lem:hyper}
  In a hyper-extensive category, given a coproduct $A =
  \coprod_{n<\omega} A_n$ with injections $a_n: A_n \to A$, the
  subobjects
  \[
    \bar a_k = [a_0, [a_n]_{n\geq k}]: A_0 + \coprod_{n \geq k} A_n \to A\qquad
    (k\geq 1)
  \]
  have the intersection $a_0: A_0 \to A$.
\end{lemma}
\begin{proof}
  It is our task to prove that every morphism $f: B \to A$ factorizing
  through the morphisms $\bar a_k$, for every $k \geq 1$, factorizes
  through $a_0$. Due to hyper-extensivity, $f$ has the form
  $f = \coprod_{n < \omega} f_n$ for morphisms $f_n: B_n \to A_n$ with
  $B = \coprod_{n <\omega} B_n$. We now prove that since $f$ factorizes through
  $\bar a_k$ it follows that $B_n \cong 0$ for all $1 \leq n <
  k$. Indeed, for any $k > 2$, let 
  \[
    \ol A_k = A_0 + \coprod\limits_{n \geq k} A_n,
    \qquad
    \ol B_k = B_0 + \coprod\limits_{n\geq k} B_n
    \qquad
    \text{and}
    \qquad
    \ol f_k = f_0 + \coprod_{n \geq k} f_n,
  \]
  and consider for $1 \leq n < k$ the pullback squares
  
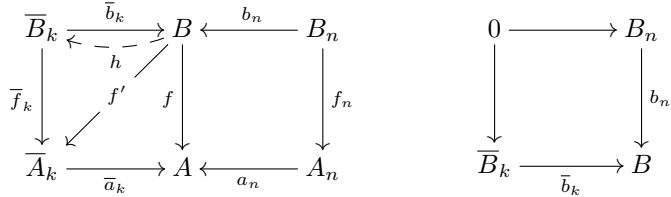
\[
    \xymatrix@+1pc{
      \ol B_k 
      \ar[r]^-{\ol b_k} 
      \ar[d]_{\ol f_k}
      & 
      B 
      \ar[d]_f
      \ar[ld]|-*+{\labelstyle f'}
      \ar@/^1.5ex/@{-->}[l]^-{h}
      & 
      B_n 
      \ar[l]_{b_n}
      \ar[d]^{f_n}
      \\
      \ol A_k 
      \ar[r]_-{\ol a_k}
      & A & A_n \ar[l]^-{a_n}
      }
      \qquad\qquad
      \xymatrix@+1pc{
        0 \ar[d] \ar[r] & B_n \ar[d]^{b_n} \\
        \ol B_k \ar[r]_-{\ol b_k} & B
      }
  \]
  Since $f$ factorizes through $\ol a_k$ we have the diagonal morphism
  $f'$ on the left such that the triangle below it commutes. Using
  the universal property of the left-hand pullback we then obtain a unique $h:
  B\to \ol B_k$ such that $\ol f_k \cdot  = f'$ and $\ol b_k
  \cdot h = \id_B$. This shows that the coproduct injection $\ol
  b_k$ is a split epimorphism, and since it is also a monomorphism by
  extensivity, we see that $\ol b_k$ is an isomorphism. Now consider the
  pullback on the right above, which expresses that the coproduct
  injections $b_n$ and $\ol b_k$ are disjoint. Since the morphism at
  the bottom is an isomorphism so is the morphism at the top, whence
  $B_n \cong 0$ for all $1 \leq n < k$.  

  Since this holds for every $k \geq 1$, we have shown that
  $B_n \cong 0$ for all $n \geq 1$. Thus, we obtain $B \cong B_0$ as
  desired.
\end{proof}
\fi
\begin{theorem}\label{thm:coprod}
  The free cia on $Y$ is
  \[
    CY = H^*Y + \nu H
  \]
  with algebra structure 
  $\phi_Y + \ter^{-1}: H(H^*Y + \nu H) \cong HH^*Y + H(\nu H) \to H^*Y
  +\nu H$.
\end{theorem} 
\iffull
\begin{proof}
  In view of Example~\ref{ex:cia} it suffices to prove that the terminal
  coalgebra for $Y+H(-)$ is $H^*Y + \nu H$ with the following
  coalgebra structure
  \[
    H^*Y + \nu H \xrightarrow{\sigma_Y + \ter} Y + HH^* Y + H(\nu H)
    \cong Y + H(H^* Y + \nu H).
  \]
  This means that for a given
  coalgebra $e: X \to Y + HX$ there exists precisely one morphism $h:
  X \to H^*Y + \nu H$ such that the following square commutes: 
  \begin{equation}\label{diag:mor}
    \vcenter{
    \xymatrix@C+1pc{
      X \ar[r]^-h \ar[d]_e
      &
      H^* Y + \nu H
      \ar[d]^{\sigma_Y + \ter}
      \\
      Y + HX 
      \ar[r]_-{Y + Hh}
      &
      Y + H(H^*Y + \nu H)
    }}
  \end{equation}
  
  (a)~Uniqueness. We define countably many pairwise disjoint
  subobjects of $X$ and prove that $h$ is uniquely determined by the
  given equation morphism $e$ on each of them. That will conclude the
  proof of uniqueness since we will see that $X$ is the coproduct of
  all of those subobjects. To start, we put
  \[
    X_0 = X\qquad\text{and}\qquad e_0 = e,
  \]
  and denote the coproduct injections of $Y + HX$ by 
  \[
    HX \xrightarrow{i_0} Y + HX
    \qquad\text{and}\qquad
    Y \xrightarrow{\ol i_0} Y + HX.
  \]
  Next form the pullbacks of $e$ along these injections:
  \begin{equation}\label{diag:pull}
    \vcenter{
    \xymatrix{
      X_1 
      \ar[r]^-{i_1}
      \ar[d]_{e_1}
      &
      X_0
      \ar[d]^{e_0}
      &
      \ol X_1
      \ar[l]_-{\ol i_1}
      \ar[d]^{\ol e_1}
      \\
      HX \ar[r]_-{i_0} 
      & 
      Y + HX 
      &
      Y 
      \ar[l]^-{\ol i_0}
    }}
  \end{equation}
  By extensivity, $X = X_1 + \ol X_1$ with injections $i_1$ and $\ol i_1$. The component $\ol h_1 := h
  \cdot \ol i_1$ of $h$ at $\ol X_1$ is determined by $e$ as follows
  \[
    h \cdot \ol i_1 = \left(
    \ol X_1 \xrightarrow{\ol e_1}
    Y 
    \xrightarrow{j_0}
    H^*Y
    \xrightarrow{\inl}
    H^* Y + \nu H
    \right).
  \]
  This follows from the commutative diagram below (note that
  from~\eqref{diag:tri} we see that the right-hand and lower arrows
  compose to
  $Y \xrightarrow{j_0} H^*Y \xrightarrow{\inl} H^*Y + \nu H$):
  \begin{equation}\label{diag:ind1}
    \vcenter{
    \xymatrix{
      \ol X_1 
      \ar[r]^-{\ol i_1}
      \ar[d]_{\ol e_1}
      &
      X
      \ar[rr]^-h
      \ar[d]^e
      &&
      H^*Y + \nu H
      \\
      Y 
      \ar[r]_-{\ol i_0}
      &
      Y + HX
      \ar[rr]_-{Y + Hh}
      &&
      Y + H(H^* Y + \nu H)
      \ar[u]_{(\sigma_Y + \ter)^{-1} = \sigma_Y^{-1} + \ter^{-1}}
      }
    }
  \end{equation}
  In order to analyze the complementary coproduct component $h\o i_1$,
  we form the pullbacks of $e_1$ along the coproduct injections of
  $HX_0 = HX_1 + H\ol X_1$:
  \[
    \xymatrix{
      X_2
      \ar[r]^-{i_2}
      \ar[d]_{e_2}
      &
      X_1
      \ar[d]^{e_1}
      &
      \ol X_2
      \ar[d]^{\ol e_2}
      \ar[l]_-{\ol i_2}
      \\
      HX_1
      \ar[r]_-{Hi_1}
      &
      HX_0
      &
      H\ol X_1
      \ar[l]^-{H\ol i_1}
      }
  \]
  Then $X_1 = X_2 + \ol X_2$ and the component $\ol h_2 = h \o i_1 \o
  \ol i_2$ of $h$ at $\ol X_2$ is determined by $e$ as follows:
  \[
    h \o i_1 \o \ol i_2 = \left(
      \ol X_2
      \xrightarrow{\ol e_2}
      H\ol X_1
      \xrightarrow{H\ol e_1}
      HY
      \xrightarrow{j_1}
      H^* Y
      \xrightarrow{\inl}
      H^*Y + \nu H
      \right).
  \]
  This follows from the commutative diagram below
  (from~\eqref{diag:tri} we see that the right-hand and lower arrows
  compose to $HY \xrightarrow{j_1} H^*Y \xrightarrow{\inl} H^* Y +
  \nu H$):
  \begin{equation}\label{diag:ind2}
    \vcenter{
    \let\objectstyle=\labelstyle
    \xymatrix@C+1pc{
      \ol X_2
      \ar[r]^-{\ol i_2}
      \ar[d]_{\ol e_2}
      &
      X_1
      \ar[r]^-{i_1}
      \ar[d]_{e_1}
      &
      X
      \ar[rr]^h
      \ar[d]_e
      &&
      H^*Y + \nu H
      \\
      H\ol X_1 
      \ar[dd]_{H \ol e_1}
      \ar[r]^-{H\ol i_1}
      & 
      HX
      \ar[r]^-{i_0}
      \ar[d]_{He}
      &
      Y + HX
      \ar[rr]^-{Y + Hh}
      \ar[d]_{Y + He}
      &&
      Y + H(H^*Y + \nu H)
      \ar[u]^{(\sigma_Y + \ter)^{-1}}
      \\
      & 
      H(Y + HX)
      \ar@{=}[d]
      \ar[r]^-\inr
      &
      Y+H(Y + HX)
      \ar@{=}[d]
      \ar[rr]^-{Y + H(Y + Hh)}
      &&
      Y + H(Y + H(H^*Y + \nu H))
      \ar[u]^{Y + H(\sigma_Y + \ter)^{-1}}
      \\
      HY
      \ar[ru]^-{H\ol i_0}
      \ar[r]_-{\inl}
      &
      HY + HHX
      \ar[r]_-{\inr}
      &
      Y + HY + HHX
      \ar[rr]_-{Y + HY + HHh}
      &&
      Y + HY + HH(H^*Y + \nu H)
      \ar@{=}[u]
    }
    }
  \end{equation}
  We continue this process recursively: given a coproduct
  $X_n \xrightarrow{i_n} X_{n-1} \xleftarrow{\ol i_n} \ol X_n$
  and a morphism $e_n: X_n \to HX_{n-1}$ we form its pullbacks along
  the coproduct injection of $HX_{n-1} = HX_n + H\ol X_n$:
  \begin{equation}\label{diag:+}
    \vcenter{
      \xymatrix{
        X_{n+1} 
        \ar[r]^-{i_{n+1}} 
        \ar[d]_{e_{n+1}}
        & 
        X_n
        \ar[d]_{e_n}
        & 
        \ol X_{n+1} 
        \ar[l]_-{\ol i_{n+1}} 
        \ar[d]^{\ol e_{n+1}}
        \\
        HX_n
        \ar[r]_-{Hi_n}
        &
        HX_{n-1}
        &
        H\ol X_n
        \ar[l]^-{H\ol i_n}
      }
    }
  \end{equation}
  Since compositions of coproduct injections are always coproduct injections,
  we obtain coproduct injections
  \begin{equation}\label{eq:inj}
    \ol i_{n+1}^* = \left(
      \ol X_{n+1}
      \xrightarrow{\ol i_{n+1}}
      X_n
      \xrightarrow{i_n}
      X_{n+1}
      \xrightarrow{i_{n-1}}
      \cdots
      \xrightarrow{i_1}
      X
      \right)
      \qquad (n < \omega)
  \end{equation}
  and morphisms
  \begin{equation}\label{eq:mor}
    \wh e_{n+1} = \left(
      \ol X_{n+1}
      \xrightarrow{\ol e_{n+1}}
      H\ol X_n 
      \xrightarrow{H\ol e_{n}}
      H^2\ol X_{n-1}
      \xrightarrow{H^2\ol e_{n-1}}
      \cdots
      \xrightarrow{H^{n} \ol e_1}
      H^{n} Y
    \right)
    \qquad
    (n < \omega).
  \end{equation}
  The component $\ol h_{n+1} := (\ol X_{n+1} \xrightarrow{\ol i_{n+1}^*} X
  \xrightarrow{h} H^*Y + \nu H)$ of $h$ at $\ol X_{n+1}$ is determined
  by $e$ via the commutativity of the following square
  \begin{equation}\label{diag:deth}
    \vcenter{
      \xymatrix{
        \ol X_{n+1} 
        \ar[rr]^-{\ol i_{n+1}^*}
        \ar[d]_{\wh e_{n+1}}
        &&
        X
        \ar[d]^h 
        \\
        H^n Y \ar[r]_-{j_n}
        &
        H^* Y \ar[r]_-{\inl}
        &
        H^*Y + \nu H
      }
    }
  \end{equation}
  The proof is by an obvious inductive continuation of the
  diagrams~\eqref{diag:ind1} and~\eqref{diag:ind2}. 
  Observe also that by composing pullback squares we obtain the
  following pullback:
  \begin{equation}\label{diag:bigpb}
    \vcenter{
      \xymatrix@C-.2pc{
        \ol X_n
        \ar[r]^-{\ol i_n}
        \ar[d]_{\ol e_n}
        &
        X_{n-1}
        \ar[r]^-{i_{n-1}}
        \ar[d]_{e_{n-1}}
        &
        X_{n-2}
        \ar[r]^-{i_{n-2}}
        \ar[d]_{e_{n-2}}
        &
        \cdots
        \ar[r]^-{i_3}
        \ar@{}[d]|{\objectstyle\cdots}
        &
        X_2
        \ar[r]^-{i_2}
        \ar[d]_{e_2}
        &
        X_1
        \ar[r]^-{i_1}
        \ar[d]_{e_1}
        &
        X_0 = X
        \ar[d]^e
        \ar@{<-} `u[l] `[llllll]_-{\ol i_n^*}
        \\
        H\ol X_{n-1}
        \ar[r]_-{H\ol i_{n-1}}
        &
        HX_{n-2}
        \ar[r]_-{Hi_{n-2}}
        &
        HX_{n-3}
        \ar[r]_-{Hi_{n-3}}
        &
        \cdots
        \ar[r]_-{Hi_2}
        &
        HX_1
        \ar[r]_-{Hi_1}
        &
        HX_0
        \ar[r]_-{i_0}
        \ar@{<-} `d[l] `[lllll]^-{H\ol i_{n-1}^*}
        &
        Y + HX
      }
    }
  \end{equation}
  Now the coproduct injections in~\eqref{eq:inj} are clearly pairwise
  disjoint. Therefore, by hyper-extensivity, we have a coproduct
  injection $[\ol i_{n+1}^*]_{n<\omega}$ which we denote by 
  \[
    \ol X_\infty \xrightarrow{\ol i_\infty} X\qquad\text{for}\qquad
    \ol X_\infty := \coprod_{n < \omega} \ol X_{n+1},
  \]
  and $h \o \ol i_\infty$ is, as proved by~\eqref{diag:deth},
  determined by $e$. Now let $i_\infty: X_\infty \to X$ be the
  complementary coproduct component, i.e., we have the coproduct
  \[
    \ol X_\infty \xrightarrow{\ol i_\infty} 
    X 
    \xleftarrow{i_\infty} X_\infty.
  \]
  Since the pullbacks~\eqref{diag:bigpb} have pairwise
  disjoint coproduct injections as their upper arrows, they form
  together the pullback on the left below:
  \begin{equation}\label{diag:inf}
    \vcenter{
    \xymatrix{
      \ol X_\infty = \ol X_1 + \ol X_2 + \ol X_3 + \cdots
      \ar[r]^-{\ol i_\infty}
      \ar[d]_{\coprod \ol e_n}
      &
      X
      \ar[d]^e
      &
      X_\infty
      \ar[l]_-{i_\infty}
      \ar[dd]^{e_\infty}
      \\
      \qquad\qquad
      Y + H\ol X_1 + H\ol X_2 + \cdots
      \ar@{=}[d]
      \ar[r]^-{Y+ H\ol i_\infty}
      &
      Y + HX 
      \ar@{=}[d]
      \\
      Y + H\ol X_\infty
      \ar[r]_-{\inl}
      &
      Y + H\ol X_\infty + HX_\infty
      &
      HX_\infty
      \ar[l]^-{\inr}
      \ar[lu]_-{\inr \o H i_\infty}
    }
    }
  \end{equation}
  By extensivity, we obtain a morphism
  $e_\infty: X_\infty \to HX_\infty$ complementary to
  $\coprod \ol e_n$. This morphism is the structure of an
  $H$-coalgebra on $X_\infty$. Thus, in order to finish the proof of
  unicity of $h: X \to H^* Y + \nu H$ we only have to verify that the
  remaining coproduct component $h \o i_\infty$ is determined by
  $e$. To this end it suffices to prove that $h \o i_\infty$
  factorizes through the coproduct injections
  $\inr: \nu H \to H^*Y + \nu H$. Indeed, given a factorization
  $k: X_\infty \to \nu H$ such that the following square commutes:
  \begin{equation}\label{diag:k}
    \vcenter{
    \xymatrix{
      X_\infty
      \ar[r]^-{i_\infty} \ar[d]_k 
      & 
      X
      \ar[d]^h 
      \\
      \nu H
      \ar[r]_-{\inr} 
      &
      H^*Y + \nu H
    }
    }
  \end{equation}
  it follows that $k$ is the unique(!) $H$-coalgebra morphism from
  $e_\infty$ to $\ter$, i.e., the square below commutes:
  \begin{equation}\label{diag:kmor}
    \vcenter{
      \xymatrix{
        X_\infty \ar[r]^-k \ar[d]_{e_\infty} & \nu H \ar[d]^\ter \\
        HX_\infty \ar[r]_-{Hk} & H(\nu H)
      }
    }
  \end{equation}
  To see this consider the diagram below:
  \[
    \xymatrix{
      X \ar[ddd]_e \ar[rrr]^-h 
      &&& 
      H^*Y + \nu H 
      \ar[ddd]^{\sigma_Y + \ter}
      \\
      &
      X_\infty
      \ar[d]_{e_\infty} \ar[r]^-k 
      \ar[lu]_{i_\infty}
      \ar@{}[ru]|{\eqref{diag:k}}
      \ar@{}[ld]|{\eqref{diag:inf}}
      &
      \nu H
      \ar[ru]_-\inr
      \ar[d]^{\ter}
      \\
      &
      HX_\infty
      \ar[ld]_-{\inr \o Hi_\infty}
      \ar[r]^-{Hk}
      \ar@{}[rd]|{\eqref{diag:k}}
      &
      H(\nu H)
      \ar[rd]^-{\inr}
      \\
      Y + HX 
      \ar[rrr]_-{Y + Hh}
      &&&
      Y + HH^*Y + H(\nu H)
    }
  \]
  Its outside is the square~\eqref{diag:mor}, and all inner parts,
  except perhaps the inner square, commute. Thus, that square also
  commutes since the coproduct injection $\inr$ is monomorphic (see
  Definition~\ref{dfn:hyper}).

  The proof that $h \o i_\infty$ factorizes throught
  $\inr: \nu H \to H^*Y + \nu H$ is based on Lemma~\ref{lem:hyper},
  which shows that $\inr: \nu H \to H^*Y + \nu H$ is the intersection
  of the following coproduct injections
  \[
    b_k= \left(
      \coprod_{n \geq k} H^n Y + \nu H 
      \xrightarrow{[j_n]_{n \geq k} + \nu H} 
      H^*Y + \nu H
      \right)
    \qquad
    (k \geq 1).
  \]
  Thus, we only need to verify that $h\o i_\infty$ factorizes through every
  $b_k$. For $k = 1$ consider the diagram below:
  \[
    \xymatrix@C+0.5pc{
      X_\infty
      \ar[r]^-{i_\infty}
      \ar[dd]_{e_\infty}
      &
      X
      \ar[d]_e
      \ar[r]^-h
      &
      H^*Y + \nu H
      \ar@{<-} `r[rd] `[dd]^{b_1} [dd]
      &
      \\
      &
      Y + HX
      \ar[r]^-{Y+ Hh}
      &
      Y + HH^*Y + H(\nu H)
      \ar[u]_{(\sigma_Y + \ter)^{-1} = \sigma_Y^{-1} + t^{-1}}
      &
      \\
      HX_\infty
      \ar[r]^-{H i_\infty}
      &
      HX \ar[u]^\inr
      \ar[r]_-{Hh}
      &
      HH^*Y + H(\nu H)
      \ar[u]_{\inr}
      &
    }
  \]
  The right-hand part commutes by~\eqref{diag:tri}, for the left-hand
  part see the upper right-hand part of~\eqref{diag:inf}, the upper
  middle part commutes by~\eqref{diag:mor} and the remaining lower
  middle part trivially commutes.
  
  Given a factorization of $h \o i_\infty$ through $b_k$ via $f$, then
  $H(h \o i_\infty)$ factorizes through $Hb_k$ via $Hf$. Using this we
  conclude that $h \o i_\infty$ factorizes through $b_{k+1}$ using the
  diagram below:
    \[
    \xymatrix{
      X_\infty
      \ar[r]^-{i_\infty}
      \ar[ddd]_{e_\infty}
      &
      X
      \ar[d]_e
      \ar[r]^-h
      &
      H^*Y + \nu H
      &
      \\
      &
      Y + HX
      \ar[r]^-{Y+ Hh}
      &
      Y + HH^*Y + H(\nu H)
      \ar[u]_{(\sigma_Y + \ter)^{-1} = \sigma_Y^{-1} + \ter^{-1}}
      &
      \\
      &
      HX \ar[u]^\inr
      \ar[r]_-{Hh}
      &
      HH^*Y + H(\nu H)
      \ar[u]_{\inr}
      &
      \coprod\limits_{n\geq k+1} H^n Y + \nu H
      \ar  `u[uul]_{b_{k+1}} [uul]
      \\
      HX_\infty
      \ar[ru]^-{Hi_\infty}
      \ar[r]_-{Hf}
      &
      H\left(\coprod\limits_{n\geq k} H^n Y + \nu H\right)
      \ar[ru]_-{Hb_k}
      \ar[r]_-{\cong}
      &
      \coprod\limits_{n\geq k} H^{n+1} Y + H(\nu H)
      \ar[ru]_{\id + \ter^{-1}}
    }
  \]
  All its inner parts, except perhaps the right-hand one clearly
  commute. For the remaining right-hand part, we consider the
  components of the coproduct in its lower left-hand corner
  separately: the right-hand component with domain $H(\nu H)$ has
  $\ter^{-1}$ on both paths. We further consider the components of
  $H(\coprod_{n \geq k} H^nY)$ with the help of the diagram below:
  \[
    \xymatrix{
      H(H^n Y) 
      \ar[d]_(.4){H\inj_n}
      \ar[rd]^{Hj_n}
      \ar@{=}[rr]
      &&
      H^{n+1} Y
      \ar[r]^-{\inj_n}
      \ar[dd]_{j_{n+1}}
      \ar[rd]_{\inj_{n+1}}
      &
      \coprod\limits_{n\geq k} H^{n+1}Y + \nu H
      \ar[d]^{\id + \ter^{-1}}
      \\
      H\left(\coprod\limits_{n\geq k} H^n Y + \nu H\right)
      \ar[d]_{Hb_k}
      &
      HH^*Y 
      \ar[ld]_\inl
      \ar[d]^\inr
      &&
      \coprod\limits_{n\geq k} H^n Y + \nu H
      \ar[dd]^{b_{k+1}}
      \\
      HH^* Y + H(\nu H)
      \ar[d]_\inr
      &
      Y + HH^*Y 
      \ar[ld]_\inl
      \ar@<-3pt>[r]_-{\sigma_Y^{-1}}
      &
      H^*Y
      \ar@<-3pt>[l]_-{\sigma_Y}
      \ar[rd]^{\inl}
      \\
      Y + HH^*Y + H(\nu H)
      \ar[rrr]_-{\sigma_Y^{-1} + \ter^{-1}}
      &&&
      H^*Y + \nu H
    }
  \]
  Its upper central part commutes by~\eqref{diag:tri}, the left-hand
  triangle commutes by the definition of $b_k$ and the right-hand
  rhombus by the definition of $b_{k+1}$; all other inner parts
  clearly commute.

  We conclude that $h$ is unique since it is equal to 
  \[
    X = \coprod_{n\geq 1} \ol X_n + X_\infty 
    \xrightarrow{[\ol h_n] + k} H^*Y + \nu H.
    \
  \]

  (b)~Existence: For the given coalgebra $e$ we define $i_n$,
  $\ol i_n$, $e_n$ and $\ol e_n$ by~\eqref{diag:pull}
  and~\eqref{diag:+}, and we also define
  $e_\infty: X_\infty \to HX_\infty$ by~\eqref{diag:inf} where
  $X = \ol X_\infty + X_\infty$ with
  $\ol X_\infty = \coprod_{n\geq 1} \ol X_n$. We furthermore use
  notations~\eqref{eq:inj} and~\eqref{eq:mor}.

  Define $k: X_\infty \to \nu H$ by~\eqref{diag:k} and for all $n \geq 1$ put
  \begin{equation}\label{eq:barh}
    \ol h_n = \left(
      \ol X_n \xrightarrow{\wh e_n} H^nY \xrightarrow{j_n} H^* Y
    \right).
  \end{equation}
  We prove that $[\ol h_n] + k: X \to H^*Y + \nu H$
  is a coalgebra morphism for $Y + H(-)$, i.e., the
  square below commutes:
  \[
    \xymatrix@C+1pc{
      X
      \ar@{=}[r]
      \ar[d]_e
      &
      \coprod_{n\geq 1} \ol X_n + X_\infty
      \ar[rr]^-{[\ol h_n] + k}
      \ar[d]_{\coprod_{n\geq 1} \ol e_n + e_\infty}
      &&
      H^* Y + \nu H
      \ar[d]^{\sigma_Y + \ter}
      \\
      Y + HX 
      \ar@{=}[r] 
      &
      Y + \coprod_{n \geq 1} H\ol X_n + HX_\infty
      \ar[rr]_-{Y + [H\ol h_n] + Hk}
      &&
      Y + H H^*Y + H(\nu H)
    }
  \]
  Its right-hand coproduct component with domain $X_\infty$ is the
  square~\eqref{diag:k} defining $k$ by the commutativity of the
  right-hand part of~\eqref{diag:inf}.

  Let us verify that the coproduct components with domain $\ol X_n$
  commute. We proceed by induction on $n$. For the base case we
  obtain the following commutative diagram (for the right-hand
  triangle see~\eqref{diag:tri}, and for the left-hand one see~\eqref{diag:pull}):
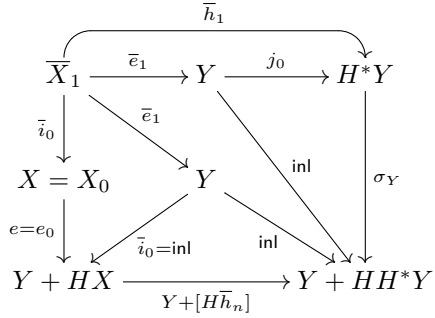
  \[
    \xymatrix{
      \ol X_1
      \ar[r]^-{\ol e_1}
      \ar[rd]^-{\ol e_1}
      \ar[d]_{\ol i_0}
      &
      Y
      \ar[r]^-{j_0}
      \ar[rdd]^-{\inl}
      &
      H^* Y
      \ar[dd]^{\sigma_Y}
      \ar@{<-} `u[l] `[ll]_-{\ol h_1}
      \\
      X=X_0
      \ar[d]_{e = e_0}
      &
      Y
      \ar[rd]_-{\inl}
      \ar[ld]^-{\ol i_0 = \inl}
      \\
      Y + HX
      \ar[rr]_-{Y + [H\ol h_n]}
      &&
      Y + HH^*Y
    }
  \]
  For the induction step with $n > 1$ consider the diagram below:
  
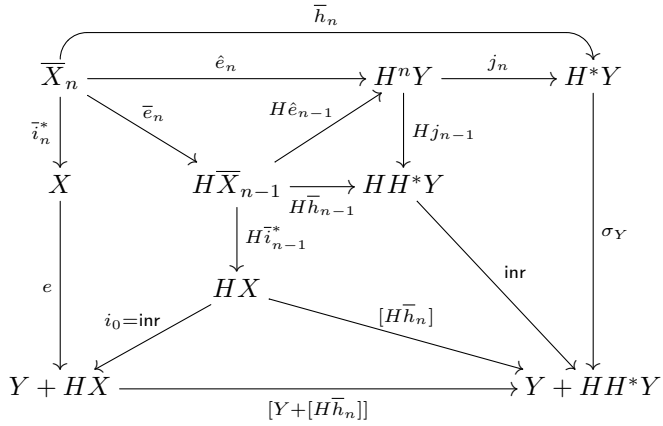
\[
    \xymatrix{
      \ol X_n
      \ar[rr]^-{\hat e_n}
      \ar[rd]^-{\ol e_n}
      \ar[d]_{\ol i_n^*}
      &&
      H^n Y
      \ar[r]^-{j_n}
      \ar[d]^{Hj_{n-1}}
      &
      H^* Y
      \ar@{<-} `u[l] `[lll]_-{\ol h_n}
      \ar[ddd]^{\sigma_Y}
      \\
      X
      \ar[dd]_e
      &
      H\ol X_{n-1}
      \ar[ru]^-{H\hat e_{n-1}}
      \ar[r]_-{H\ol h_{n-1}}
      \ar[d]^{H\ol i_{n-1}^*}
      &
      HH^*Y
      \ar[rdd]^-\inr
      \\
      &
      HX
      \ar[ld]_-{i_0 = \inr}
      \ar[rrd]^-{[H\ol h_n]}
      \\
      Y + HX
      \ar[rrr]_-{[Y + [H\ol h_n]]}
      &&&
      Y + HH^*Y
    }
  \]
  The upper part and the middle triangle under it commute by~\eqref{eq:barh},
  the upper left-hand triangle follows immediately from~\eqref{eq:mor}. The
  right-hand part commutes by~\eqref{diag:tri}, and the left-hand part
  is the outside of~\eqref{diag:bigpb}. The remaining parts clearly commute.
\end{proof}
\else
\begin{proof}[Proofsketch]
  In view of Example~\ref{ex:cia} it suffices to prove that the terminal
  coalgebra for $Y+H(-)$ is $H^*Y + \nu H$ with the following
  coalgebra structure
  \[
    H^*Y + \nu H \xrightarrow{\sigma_Y + \ter} Y + HH^* Y + H(\nu H)
    \cong Y + H(H^* Y + \nu H).
  \]
  That means that for a given coalgebra $e: X \to Y + HX$ there exists
  precisely one coalgebra morphism $h: X \to H^*Y + \nu H$.  This
  morphism $h$ is defined by an iterative construction using pullbacks
  and (hyper-)extensivity that we now explain.

  Let $X_0 = X$ and $e_0 = e$ and denote the coproduct injections of
  $Y + HX$ by $i_0: HX \to Y + HX$ and
  $\ol i_0: Y \to Y + HX$.  Next form the pullbacks
  of $e$ along these injections:
  \begin{equation}\label{diag:pull}
    \vcenter{
    \xymatrix{
      X_1 
      \ar[r]^-{i_1}
      \ar[d]_{e_1}
      &
      X_0
      \ar[d]^{e_0}
      &
      \ol X_1
      \ar[l]_-{\ol i_1}
      \ar[d]^{\ol e_1}
      \\
      HX \ar[r]_-{i_0} 
      & 
      Y + HX 
      &
      Y 
      \ar[l]^-{\ol i_0}
    }}
  \end{equation}
  By extensivity, $X = X_1 + \ol X_1$ with injections $i_1$ and
  $\ol i_1$. The component $\ol h_1 := h \cdot \ol i_1$ of $h$ at
  $\ol X_1$ is defined by
  \[
    h \cdot \ol i_1 = \left(
    \ol X_1 \xrightarrow{\ol e_1}
    Y 
    \xrightarrow{j_0}
    H^*Y
    \xrightarrow{\inl}
    H^* Y + \nu H
    \right).
  \]
    In order to analyze the complementary coproduct component $h\o i_1$,
  we form the pullbacks of $e_1$ along the coproduct injections of
  $HX_0 = HX_1 + H\ol X_1$:
  \[
    \xymatrix{
      X_2
      \ar[r]^-{i_2}
      \ar[d]_{e_2}
      &
      X_1
      \ar[d]^{e_1}
      &
      \ol X_2
      \ar[d]^{\ol e_2}
      \ar[l]_-{\ol i_2}
      \\
      HX_1
      \ar[r]_-{Hi_1}
      &
      HX_0
      &
      H\ol X_1
      \ar[l]^-{H\ol i_1}
      }
  \]
  Then $X_1 = X_2 + \ol X_2$ and the component $\ol h_2 = h \o i_1 \o
  \ol i_2$ of $h$ at $\ol X_2$ is defined by
  \[
    h \o i_1 \o \ol i_2 = \left(
      \ol X_2
      \xrightarrow{\ol e_2}
      H\ol X_1
      \xrightarrow{H\ol e_1}
      HY
      \xrightarrow{j_1}
      H^* Y
      \xrightarrow{\inl}
      H^*Y + \nu H
      \right).
  \]
    We continue this process recursively: given a coproduct
  $X_n \xrightarrow{i_n} X_{n-1} \xleftarrow{\ol i_n} \ol X_n$
  and a morphism $e_n: X_n \to HX_{n-1}$ we form its pullbacks along
  the coproduct injection of $HX_{n-1} = HX_n + H\ol X_n$:
  \begin{equation}\label{diag:+}
    \vcenter{
      \xymatrix{
        X_{n+1} 
        \ar[r]^-{i_{n+1}} 
        \ar[d]_{e_{n+1}}
        & 
        X_n
        \ar[d]_{e_n}
        & 
        \ol X_{n+1} 
        \ar[l]_-{\ol i_{n+1}} 
        \ar[d]^{\ol e_{n+1}}
        \\
        HX_n
        \ar[r]_-{Hi_n}
        &
        HX_{n-1}
        &
        H\ol X_n
        \ar[l]^-{H\ol i_n}
      }
    }
  \end{equation}
  Since compositions of coproduct injections are always coproduct injections, 
  we obtain coproduct injections
  \begin{equation}\label{eq:inj}
    \ol i_{n+1}^* = \left(
      \ol X_{n+1}
      \xrightarrow{\ol i_{n+1}}
      X_n
      \xrightarrow{i_n}
      X_{n+1}
      \xrightarrow{i_{n-1}}
      \cdots
      \xrightarrow{i_1}
      X
      \right)
      \qquad (n < \omega)
  \end{equation}
  and morphisms
  \begin{equation}\label{eq:mor}
    \wh e_{n+1} = \left(
      \ol X_{n+1}
      \xrightarrow{\ol e_{n+1}}
      H\ol X_n 
      \xrightarrow{H\ol e_{n}}
      H^2\ol X_{n-1}
      \xrightarrow{H^2\ol e_{n-1}}
      \cdots
      \xrightarrow{H^{n} \ol e_1}
      H^{n} Y
    \right)
    \qquad
    (n < \omega).
  \end{equation}
  The component $\ol h_{n+1} := (\ol X_{n+1} \xrightarrow{\ol i_{n+1}^*} X
  \xrightarrow{h} H^*Y + \nu H)$ of $h$ at $\ol X_{n+1}$ is defined by
  the commutativity of the following square
  \begin{equation}\label{diag:deth}
    \vcenter{
      \xymatrix{
        \ol X_{n+1} 
        \ar[rr]^-{\ol i_{n+1}^*}
        \ar[d]_{\wh e_{n+1}}
        &&
        X
        \ar[d]^h 
        \\
        H^n Y \ar[r]_-{j_n}
        &
        H^* Y \ar[r]_-{\inl}
        &
        H^*Y + \nu H
      }
    }
  \end{equation}

  Observe also that by composing pullback squares we obtain the
  following pullback:
  \begin{equation}\label{diag:bigpb}
    \vcenter{
      \xymatrix@C-.2pc{
        \ol X_n
        \ar[r]^-{\ol i_n}
        \ar[d]_{\ol e_n}
        &
        X_{n-1}
        \ar[r]^-{i_{n-1}}
        \ar[d]_{e_{n-1}}
        &
        X_{n-2}
        \ar[r]^-{i_{n-2}}
        \ar[d]_{e_{n-2}}
        &
        \cdots
        \ar[r]^-{i_3}
        \ar@{}[d]|{\objectstyle\cdots}
        &
        X_2
        \ar[r]^-{i_2}
        \ar[d]_{e_2}
        &
        X_1
        \ar[r]^-{i_1}
        \ar[d]_{e_1}
        &
        X_0 = X
        \ar[d]^e
        \ar@{<-} `u[l] `[llllll]_-{\ol i_n^*}
        \\
        H\ol X_{n-1}
        \ar[r]_-{H\ol i_{n-1}}
        &
        HX_{n-2}
        \ar[r]_-{Hi_{n-2}}
        &
        HX_{n-3}
        \ar[r]_-{Hi_{n-3}}
        &
        \cdots
        \ar[r]_-{Hi_2}
        &
        HX_1
        \ar[r]_-{Hi_1}
        &
        HX_0
        \ar[r]_-{i_0}
        \ar@{<-} `d[l] `[lllll]^-{H\ol i_{n-1}^*}
        &
        Y + HX
      }
    }
  \end{equation}
  Now the coproduct injections in~\eqref{eq:inj} are clearly pairwise
  disjoint. Therefore, by hyper-extensivity, we have a coproduct
  injection $[\ol i_{n+1}^*]_{n<\omega}$ which we denote by
  \[
    \ol X_\infty \xrightarrow{\ol i_\infty} X\qquad\text{for}\qquad
    \ol X_\infty := \coprod_{n < \omega} \ol X_{n+1},
  \]
  and $h \o \ol i_\infty$ is defined coponentwise
  by~\eqref{diag:deth}.  By hyper-extensivity we can consider the
  complementary coproduct component $i_\infty: X_\infty \to X$,
  i.e., we have the coproduct
  \[
    \ol X_\infty \xrightarrow{\ol i_\infty} 
    X 
    \xleftarrow{i_\infty} X_\infty.
  \]
  Since the pullbacks~\eqref{diag:bigpb} have pairwise
  disjoint coproduct injections as their upper arrows, they form
  together the pullback on the left below:
  \begin{equation}\label{diag:inf}
    \vcenter{
    \xymatrix{
      \ol X_\infty = \ol X_1 + \ol X_2 + \ol X_3 + \cdots
      \ar[r]^-{\ol i_\infty}
      \ar[d]_{\coprod \ol e_n}
      &
      X
      \ar[d]^e
      &
      X_\infty
      \ar[l]_-{i_\infty}
      \ar[dd]^{e_\infty}
      \\
      \qquad\qquad
      Y + H\ol X_1 + H\ol X_2 + \cdots
      \ar@{=}[d]
      \ar[r]^-{Y+ H\ol i_\infty}
      &
      Y + HX 
      \ar@{=}[d]
      \\
      Y + H\ol X_\infty
      \ar[r]_-{\inl}
      &
      Y + H\ol X_\infty + HX_\infty
      &
      HX_\infty
      \ar[l]^-{\inr}
      \ar[lu]_-{\inr \o H i_\infty}
    }
    }
  \end{equation}
  By extensivity, we obtain a morphism
  $e_\infty: X_\infty \to HX_\infty$ complementary to
  $\coprod \ol e_n$. This morphism is the structure of an
  $H$-coalgebra on $X_\infty$. Thus, we define $h \o i_\infty$ to be
  the unique coalgebra morphism from $X_\infty$ to $\nu H$. 

  One now verifies that the morphism $h: X \to H^*Y + \nu H$ so
  defined is a unique coalgebra morphism for $Y+H(-)$ as desired (see
  the full version~\cite{am17} of our paper for details).
\end{proof}
\fi

\begin{example}
  \label{ex:nat}
  \begin{enumerate}[(a)]
  \item It is well-known that the identity functor on $\Set$ has the
    free cias (equivalently, final coalgebras for $(-) + Y$)
    $TY = \N \times Y + 1$ where $\N$ is the set of natural
    numbers. It follows from Theorem~\ref{thm:coprod} that the same
    formula holds in every hyper-extensive category with a terminal
    object $1$.  To see this, one first shows that
  \[
    N := \coprod_{n< \omega} 1
    \quad
    \text{with}
    \quad
    \xymatrix@1@C+2pc{
      1 \ar[r]^-{\inj_0} 
      & 
      N 
      & 
      N \ar[l]_-{[\inj_{n+1}]_{n < \omega}}
    }
  \]
  forms a natural number object, i.e., an initial algebra for $1 +
  (-)$. Using distributivity we see that for any object $Y$ the free
  algebra $\Id^* Y$ is
  \begin{equation}\label{eq:Idstar}
   \Id^* Y 
    = 
    \coprod_{n < \omega} Y 
    \cong 
    \left(\coprod_{n < \omega} 1 \right) \times Y = N \times Y.
  \end{equation}
  Finally, we clearly have $\nu \Id = 1$. By Theorem~\ref{thm:coprod},
  we thus obtain 
  \[
    TY \cong N \times Y + 1.
  \]
\item For the above formula giving the free cia for $\Id$ on every $Y$
  it is \emph{not} sufficient that $\C$ be an extensive category. As a
  counterexample consider the category $\C = \CH$ of compact Hausdorff
  spaces. Its limits and finite coproducts are created by the
  forgetful functor into $\Set$, thus $\CH$ is extensive. However, it
  is not hyper-extensive since countable coproducts are not universal.
  For $Y = 1$ (the one point space) the formula~\eqref{eq:Idstar}
  gives an uncountable space since $\coprod_{n < \omega} 1$ is the
  Stone-\v{C}ech compactification of an infinite discrete
  space. However, in the notation of Example~\ref{ex:corec}, $T1$ is a
  countable space; for the terminal $\omega^\op$-chain
  \[
    1 \leftarrow 1 + 1 \leftarrow 1 + 1 +1 \leftarrow \cdots
  \]
  of the functor $\Id + 1$ on $\CH$ has the corresponding underlying
  chain in $\Set$. The limit in $\Set$ is countable, giving the set $N
  + 1$. The limit in $\CH$ is then a compact space on this set, in
  fact, it is the one-point compactification of the discrete space on
  $N$. Since the functor $X \mapsto X + 1$ preserves this limit, it is
  its terminal coalgebra. That means that $T1$ is countable.
\end{enumerate}
\end{example}

\begin{example}
  Extending Example~\ref{ex:nat}(a), we know that the functor $HX =
  \Sigma \times X$ on $\Set$ has the free cias $TY = \Sigma^* \times Y +
  \Sigma^\omega$, where $\Sigma^*$ and $\Sigma^\infty$ are the
  usual sets of strings (words) and sequences (streams) on $\Sigma$. 

  It follows from Theorem~\ref{thm:coprod} that the same formula holds
  in every hyper-extensive category $\C$ with finite products
  commuting with countable coproducts.  Examples of such categories
  are presheaf categories, posets, graphs and unary algebras.

  Given an object $\Sigma$ of $\C$, the functor $H X = \Sigma \times X$ has 
  the terminal coalgebra
  \[
    \Sigma^\omega = \lim\limits_{n < \omega} \Sigma^n
  \]
  which is the limit of the $\omega^{\op}$-chain of projections as
  follows:
  \[
    1 \xleftarrow{!} \Sigma \xleftarrow{\Sigma \times !} \Sigma \times
    \Sigma \xleftarrow{\Sigma \times \Sigma \times !} \Sigma \times
    \Sigma \times \Sigma \leftarrow \cdots
  \]
  The free algebras $H^* Y$ are obtained as follows: define 
  \[
    \Sigma^* = \coprod\limits_{n< \omega} \Sigma^n.
   \]
  Then $H^*Y = \Sigma^* \times Y$. Thus, according to
  Theorem~\ref{thm:coprod}, the free cia for $H$ on $Y$ is given by
  \[
    TY = \Sigma^* \times Y + \Sigma^\omega.
  \]
 
  Similarly, given another object $W$ of $\C$, the functor
  $H'X = W + \Sigma \times X$ has the free cias
  $T'Y = \Sigma^* \times (W + Y) + \Sigma^\omega$.
\end{example}
\begin{example}\label{E:ultra}
In Theorem~\ref{thm:coprod} it is not sufficient that $H$
    preserves finite coproducts. In fact, consider the ultrafilter
    functor $U: \Set \to \Set$ which assigns to every set $X$ the set
    of all ultrafilters on $X$ and to a map $f: X \to Y$ the map $Uf$
    sending an ultrafilter $\mathcal A$ on X to
    $\{B \subseteq Y \mid f^{-1}(B) \in \mathcal A\}$. It preserves
    finite coproducts and $\nu U = 1$. But for $Y$ infinite,
    $Y + U(-)$ has no fixed points; for suppose that
    $TY \cong Y + UTY$, then $TY$ must be infinite since $Y$ is so and
    therefore $|TY| < |UTY|$ contradicting the isomorphism.
\end{example}

\section{Corecursiveness vs.~Complete Iterativity}

Under Assumption~\ref{ass:3} we prove in this
section that $H$ is a cia functor, i.e., every
corecursive algebra is a cia. Let $a: HA \to A$ be a fixed algebra.
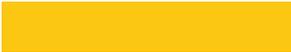
\begin{notation}\label{not:an}
  \begin{enumerate}[(1)]
  \item Define morphisms 
    \[
      a^n: H^n A \to A
    \]
    by the following induction:
    \[
      a^0 = \id_A\qquad\text{and}\qquad 
      a^{n+1} = (H^{n+1} A = HH^nA \xrightarrow{Ha^n} HA \xrightarrow{a} A).
    \]
  \item For every equation morphism $e: X \to HX + A$ we use the
    notation of the proof of Theorem~\ref{thm:coprod}, except that $Y$
    is replaced by $A$ everywhere (and the order of summands is swapped).
     Thus we use the morphisms
    \[
      i_n, \ol i_n, e_n, \ol e_n, e_\infty, i_\infty, \ol i_\infty, \wh e_n,\
      \text{and}, \ol i_n^*
    \]
    as in that proof.
\takeout{
  \item We denote by $\fin e: X_\infty \to \nu H$ the unique coalgebra
    morphism:
    \begin{equation}\label{diag:hash}
      \vcenter{
        \xymatrix{
          X_\infty
          \ar[d]_{e_\infty} \ar[r]^-{\fin e} & \nu H \ar[d]^t \\
          H X_\infty \ar[r]_-{H\fin e} & H(\nu H)
        }}
      \end{equation}
}
  \end{enumerate}
\end{notation}

\begin{construction}\label{constr:sol}
  Let $a: HA \to A$ be an algebra. Given an equation morphism $e: X \to HX + A$
  and a coalgebra-to-algebra morphism $s: X_\infty \to A$:
  \begin{equation}\label{diag:coalalg}
    \vcenter{
      \xymatrix{
        X_\infty \ar[r]^-s \ar[d]_{e_\infty} & A \\
        HX_\infty \ar[r]_-{Hs} & HA \ar[u]_a
      }
    }
  \end{equation}
  we define a morphism $\sol e_s: X \to A$ on the components of the
  coproduct $X = \left(\coprod_{n \geq 1} \ol X_n\right) + X_\infty$ (with injections
  $\ol i_n^*$, for every $n \geq 1$, and $i_\infty$) separately as follows:
  \begin{equation}\label{diag:comp}
    \vcenter{
      \xymatrix{
        \ol X_n \ar[r]^-{\wh e_n} \ar[d]_{\ol i_n^*} & H^{n-1} A \ar[d]^{a^{n-1}}\\
        X \ar[r]_-{\sol e_s} & A
      }}
      \quad
      \text{for $n \geq 1$, and}
      \quad
      \vcenter{
      \xymatrix{
        X_\infty \ar[d]_{i_\infty} \ar[rd]^s \\
        X \ar[r]_-{\sol e_s} & A
        }
      }
  \end{equation}
\end{construction}
\begin{proposition}\label{prop:uniq}
  The morphism $\sol e_s$ is a solution of $e$. Moreover, every
  solution of $e$ is of the form $\sol e_s$ for some coalgebra-to-algebra
  morphism $s$.
\end{proposition}
\iffull
\begin{proof}
  (1)~We verify the commutativity of~\eqref{diag:sol} for $\sol e_s$
  by considering the coproduct components of
  $X = \coprod_{n\geq 1} \ol X_n + X_\infty$ separately. For the
  components $\ol X_n$ we proceed by induction on $n$. For the base
  case $n = 1$ we have the diagram below:
\begin{equation}\label{diag:base}
\vcenter{
  \xymatrix{
    \ol X_1 \ar[rrr]^-{\ol e_1 = \wh e_1} \ar[ddd]_{\ol e_1}
    \ar[rd]^-{\ol i_1^* = \ol i_1}
    &&&
    A
    \ar@{=}[ld] 
    \\
    &
    X
    \ar[d]_e \ar[r]^-{\sol e_s}
    &
    A
    \\
    &
    HX + A
    \ar[r]_-{H\sol e_s + A}
    &
    HA + A \ar[u]_{[a, A]}
    \\
    A \ar[ru]_-{\ol i_0 = \inr}
    \ar@{=}[rrr]
    &&&
    A \ar@{=}[lu]
    \ar@{=}[uuu]
  }
}
\end{equation}
This commutes as follows: its left-hand part is the right-hand square
of~\eqref{diag:pull}, its upper part commutes by~\eqref{diag:comp} and
the lower and right-hand parts are trivial; since the outside also
trivially commutes so does the inner square when precomposed by $\ol
i_1^*$ as desired. 

For the induction step with $n > 1$ we consider the following diagram:
\begin{equation}\label{diag:step}
  \vcenter{
  \xymatrix@R-.2pc{
    \ol X_n
    \ar[rrrr]^-{\hat e_n}
    \ar[dddd]_{\ol e_n}
    \ar[rd]^-{\ol i_n^*}
    &&&&
    H^{n-1} A
    \ar[lld]_-{a^{n-1}}
    \\
    &
    X
    \ar[r]^-{\sol e_s}
    \ar[d]_e
    &
    A
    \\
    &
    HX + A
    \ar[r]_-{H\sol e_s + A}
    &
    HA + A
    \ar[u]^{[a,A]}
    \\
    &
    HX_0 = HX
    \ar[u]^{i_0=\inl}
    \ar[rr]_-{H\sol e_s}
    &&
    HA
    \ar[lu]^{\inl}
    \ar[luu]_a
    \\
    H\ol X_{n-1}
    \ar[ru]_-{H\ol i_{n-1}^*}
    \ar[rrrr]_-{H\hat e_{n-1}}
    &&&&
    HH^{n-2} A
    \ar@{=}[uuuu]
    \ar[lu]_-{Ha^{n-2}}
  }
}
\end{equation}
Its upper part commutes by~\eqref{diag:comp}, the left-hand part
by~\eqref{diag:bigpb}, the right-hand part commutes by the definition
of $a^{n-1}$ (see Notation~\ref{not:an}(1)), the lower part
commutes by the induction hypothesis, and the remaining two inner parts
trivially commute. That the outside commutes follows
from~\eqref{eq:mor} by an easy induction. Thus, the inner square
commutes when precomposed with $\ol i_n^*$, as desired.

Finally, for the coproduct component $X_\infty$ we consider the
following diagram:
\begin{equation}\label{diag:infcomp}
\vcenter{
  \xymatrix@-.5pc{
    X
    \ar[rrr]^-{\sol e_s}
    \ar[dddd]_e 
    &&&
    A 
    \\
    &
    X_\infty
    \ar[lu]_-{i_\infty}
    \ar[rru]_-s
    \ar[d]_{e_\infty}
    \\
    &
    HX_\infty
    \ar[d]_-{Hi\infty}
    \ar[rd]^-{Hs}
    \\
    &
    HX
    \ar[r]_-{H\sol e_s}
    \ar[ld]_-{\inl}
    & 
    HA
    \ar[ruuu]_-a
    \ar[rd]^-{\inl}
    \\
    HX + A
    \ar[rrr]_-{H\sol e_s + A}
    &&&
    HA + A
    \ar[uuuu]_{[a,A]}
  }
}
\end{equation}
Its upper part and the middle triangle commute by~\eqref{diag:comp}\footnote{Note that $HX$ is now the left-hand coproduct component while in the previous section it was the right-hand one in $Y + HX$.}, its
left-hand part is the right-hand part of~\eqref{diag:inf}, the lower
and right-hand parts trivially commute and the remaining inner part
commutes by~\eqref{diag:coalalg}. Thus, the outside commutes when
precomposed by $i_\infty$ as desired. 

(2)~Suppose that $\sol e$ is any solution of $e$, and let
$s = \sol e \cdot i_\infty: X_\infty \to A$. We will now prove that
$s$ is a coalgebra-to-algebra morphism from
$e_\infty: X_\infty \to HX_\infty$ to $a: HA \to A$ and that
$\sol e = \sol e_s$. To see the former take
Diagram~\eqref{diag:infcomp} and replace $\sol e_s$ by $\sol e$. Now
the outside commutes, and since so do all other inner parts, it
follows that the part exhibiting $s$ as coalgebra-to-algebra morphism
commutes.

To complete the proof we now show by induction on $n$ that 
\[
\sol e \cdot \ol i_n^* = a^{n-1} \cdot \wh e_n: \ol X_n \to A,
\]
cf.~\eqref{diag:comp}. It then follows that $\sol e \cdot \ol i_n =
\sol e_s \cdot \ol i_n$, and together with $\sol e \cdot i_\infty = s
= \sol e_s \cdot i_\infty$ we can conclude that $\sol e = \sol e_s$.

For the base case $n=1$ consider Diagram~\eqref{diag:base} with $\sol
e_s$ replaced by $\sol e$. Then the inner square commutes, and since all
other inner parts commute as explained in part~(1) of our proof, so
does the desired upper part.

Similary, for the induction step with $n > 1$ consider
Diagram~\eqref{diag:step} with $\sol e_s$ replaced by $\sol e$. Then
the inner square commutes, and since all other inner parts commute as
explained in part~(1) of our proof, so does the desired upper part.
This completes the proof.
\end{proof}
\fi
\begin{corollary}\label{cor:ciacorec}
  The functor $H$ is a cia functor.
\end{corollary}
Indeed, if $(A,a)$ is a corecursive $H$-algebra and $e: X \to HX + A$ is a given
equation morphism, we have a unique $s$ as
in~\eqref{diag:coalalg}. Now note that Proposition~\ref{prop:uniq}
establishes a bijective correspondence between solutions of $e$ and
coalgebra-to-algebra morphisms from $e_\infty$ to $a$, and therefore
there exists a unique solution of $e$.

\begin{example}\label{E:U}
  For the ultrafilter functor $U$ of Example \ref{E:ultra} consider
  the subfunctor $U_0$ of all $\omega$-complete ultrafilters, 
  i.e., those closed under countable intersections.  This functor preserves
  countable coproducts and $\nu U_0 = 1$. Assume that a proper class
  of measurable cardinals $n$ exists (i.e., for each $n$ we have an
  $\omega$-complete ultrafilter $P$ on a set $X$ not containing any
  subset of $X$ of less than $n$ elements). This is quite a strong
  assumption in set theory, but we make it here to derive a strong
  property of $U_0$: it is a non-accessible cia functor! Indeed, the
  latter follows from Corollary~\ref{cor:ciacorec}, and $U_0$ is not
  accessible: for every measurable cardinal $n$ it does
  not preserve the $n$-filtered colimit of all subsets $Y$ of $X$ of
  cardinality less than $n$, since $P$ lies in $U_0 X$ but not in
  $U_0 Y$ if $|Y| < n$. This is a surprising example in
  view of Theorem~\ref{T:finitary} which shows that such a complex
  example does not exist among finitary set functors.
\end{example}

Finally, note that both cias and corecursive algebras form full
subcategories of the category of all algebras for $H$. Thus
Corollary~\ref{cor:ciacorec} establishes an isomorphism of
categories between the categories of cias and corecursive algebras for
$H$.

The following proposition needs no assumptions on $H$ or the base
category except that binary coproducts exist.
\begin{proposition}
  If $H$ is a cia functor, then so is $H(-) + Y$ for every object $Y$. 
\end{proposition}
\iffull
\begin{proof}
  Let $[a,y]: HA + Y \to A$ be a corecursive algebra for $H(-) + Y$.

  (1)~The algebra $a: HA \to A$ is corecursive for $H$. Indeed, for
  every coalgebra $e: X\to HX$ we form the following coalgebra for $H(-)+Y$:
  \[
    f = (X \xrightarrow{e} HX \xrightarrow{\inl} HX + A).
  \]
  Now consider the diagram below:
  \[
    \xymatrix@C+1pc{
      X
      \ar[d]_e
      \ar[r]^-s
      & 
      A
      \\
      HX
      \ar[r]^-{Hs} 
      \ar[d]_\inl
      & 
      HA
      \ar[u]_a
      \ar[d]^{\inl}
      \\
      HX + Y 
      \ar[r]_-{Hs + Y}
      \ar@{<-} `l[u] `[uu]^f [uu]
      &
      HA + Y
      \ar `r[u] `[uu]_{[a,y]} [uu]
      }
  \]
  This shows that there is a bijective correspondence between
  coalgebra-to-algebra morphisms from $e$ to $a$ (w.r.t.~$H$) and those
  from $f$ to $[a,y]$ (w.r.t.~$H(-) + Y$). Since the former exists
  uniquely, so does the latter, hence $A$ is corecursive for
  $H$. 

  (2)~From~(1) we have by assumption that $(A,a)$ is a cia for $H$. It
  follows that $(A, [a,y])$ is a cia for $H(-) + Y$ because to give
  a cia $(A, a)$ for $H$ and a morphism $y: Y \to A$ is equivalent to
  giving a cia $(A, [a,y])$ for $H(-)+Y$, see the proof
  of~\cite[Theorem~2.10]{m_cia}.
\end{proof}
\fi
\begin{corollary}\label{C:strong}
  Let $H$ be a functor having a terminal coalgebra and preserving
  countable coproducts. Then $H(-)+Y$ is a cia functor for every object
  $Y$.
\end{corollary}

\iffull
\section{Elgot Algebras and Bloom Algebras}

Throughout this section $H$ denotes an endofunctor on a
hyper-extensive category preserving countable coproducts and having a
terminal coalgebra $\nu H$. We know that $H$ is then
\emph{iteratable}, i.e., for every $Y$ the terminal coalgebra $TY$ for
$H(-)+Y$ exists, viz.
\[
TY = \coprod_{n < \omega} H^n Y + \nu H.
\] 
This is the free cia on $Y$. According to
Corollary~\ref{cor:ciacorec}, $TY$ is also the free corecurive algebra
on $Y$.

The assignment of a free cia $TY$ to the given object $Y$ is
well-known to yield a monad $\T$; in fact, this monad is the
\emph{free completely iterative monad} on
$H$, see~\cite{aamv,m_cia}. We will not recall the notion of a
completely iterative monad here, as it is not needed in the present
paper. However, note that the unit of the monad $\T$ is given by
$\eta_Y: Y \to TY$ and the multiplication is given by freeness:
$\mu_Y: TTY \to TY$ is the unique algebra morphism extending
$\id_{TY}$ from the free cia $TTY$ on $TY$ to the cia $TY$.

The present section concerns the Eilenberg-Moore algebras for the
monad $\T$. In previous joint work with J.~Velebil~\cite{amv_elgot} we
called them complete Elgot algebras and described them as algebras 
for $H$ equipped with an operation
$\sol{(-)}$ that assigns to every equation morphism $e: X \to HX + A$
a solution $\sol e: X \to A$ satisfying two easy and well-motivated
axioms that we now recall.
\begin{notation}
  Given morphisms $e: X \to HX + Y$ and $h: Y \to Z$ we write
  \[
    h \after e = (X \xrightarrow{e} HX + Y \xrightarrow{HX + h} HX + Z).
  \]
\end{notation}
\begin{definition}
  A \emph{complete Elgot algebra} for $H$ is a triple $(A,a,\dagger)$
  where $a: HA \to A$ is an algebra and $\dagger$ is an operation
  that assigns to every equation morphism $e: X \to HX +A$ a
  solution $\sol e: X \to A$ (i.e., the square~\eqref{diag:sol} commutes) such
  that the following two properties hold:
  \begin{enumerate}[(1)]
  \item \emph{Functoriality}: for every two equation morphisms $e: X
    \to HX + A$ and $f: Y \to HY +A$ and every coalgebra morphism $h: X \to
    Y$ we have that $\sol f \cdot h = \sol e$:
    \[
      \vcenter{
      \xymatrix@R-1.5pc{
        X \ar[r]^-e \ar[dd]_h & HX + A \ar[dd]^{Hh + A} \\ \\
        Y \ar[r]_-f & HY + A
        }}
      \qquad
      \implies
      \qquad
      \vcenter{
        \xymatrix@R-1.5pc{
          X
          \ar[rd]^-{\sol e}
          \ar[dd]_h
          \\
          & A\\
          Y \ar[ru]_-{\sol f}
          }
        }
    \]
  \item \emph{Compositionality}: Given $e: X \to HY + Y$ and $f: Y \to
    HY + A$ we form the following equation morphism
    \[
       e \plus f = (X +Y \xrightarrow{[e,\inr]} 
       HX + Y \xrightarrow{HX + f} HX + HY + A \xrightarrow{\can + A}
       H(X+Y) + A);
    \]
    compositionality states that 
    \[
      \sol{(e \plus f)} \cdot \inl = \sol{(\sol f \after e)}: X \to A. 
    \]
  \end{enumerate}

  A morphism of complete Elgot algebras from $(A, a, \dagger)$ to
  $(B, b, \ddagger)$ is a morhism $h: A \to B$ \emph{preserving
    solutions}, i.e., for every $e: X \to HX + A$ the following triangle
  commutes:
  
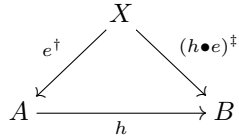
\[
    \xymatrix{
      &
      X
      \ar[ld]_{\sol e} \ar[rd]^{(h \after e)^\ddagger}
      \\
      A
      \ar[rr]_-h 
      && 
      B
    }
  \]
\end{definition}
Note that every morphism of complete Elgot algebras is an $H$-algebra
morphism from $(A,a)$ to $(B,b)$~\cite[Lemma~5.2]{amv_elgot}.
Further recall from \emph{loc.~cit.}~that every cia for $H$ is a
complete Elgot algebra; in fact, one readily proves that the operation
assigning to a given equation morphism its unique solution satisfies
functoriality and compositionality. Further examples of complete Elgot
algebras are algebras on cpos with continuous algebra structure and
algebras on non-empty complete metric spaces with contracting algebra
structure~\cite{amv_elgot}.

The following result holds for every iteratable endofunctor $H$ on a
category with binary coproducts. 
\begin{theorem}[\cite{amv_elgot}]
  The category of Eilenberg-Moore algebras for $\T$ is isomorphic to
  the category of complete Elgot algebras and their morphisms.
\end{theorem}
Of course, in the light of Corollary~\ref{cor:ciacorec}, the monad
$\T$ is also the monad of free corecursive algebras. For an accessible
endofunctor on a locally presentable category we have described the
Eilenberg-Moore algebras for that monad in~\cite{ahm14}. We now recall
the definition.
\begin{definition}
A \emph{Bloom algebra} is a triple $(A,a,\dagger)$ where
$a:HA\rightarrow A$ is an $H$-algebra and $\dagger$ is an operation
assigning to every coalgebra $e:X\rightarrow HX$ a
coalgebra-to-algebra morphism $\sol e:X\rightarrow A$ so that
$\dagger$ is functorial. This means that we obtain a functor
$$\dagger:\mathsf{Coalg}\, H\rightarrow \C/A.$$
More explicitly, given a coalgebra morphism $h$ from $(X,e)$ to $(Y, f)$ we have $\sol f \cdot h = \sol e$:
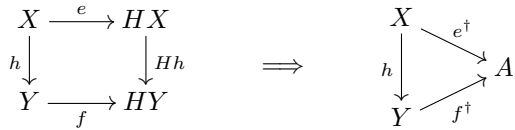
\[
  \vcenter{
    \xymatrix@R-1.5pc{
      X\ar[r]^-{e}\ar[dd]_{h}
      &
      HX\ar[dd]^{Hh}
      \\ \\
      Y \ar[r]_-{f} 
      & 
      HY
    }
  }
  \qquad
  \implies
  \qquad
  \vcenter{
    \xymatrix@R-1.5pc{
      X
      \ar[rd]^-{\sol e}
      \ar[dd]_h
      \\
      & A\\
      Y \ar[ru]_-{\sol f}
    }
  }
\]
\end{definition}
Bloom algebras form a category together with solution preserving
algebra morphisms (defined
completely analogously as for complete Elgot algebras).

We will now prove that under our current assumption Bloom algebras and
complete Elgot algebras are the same concept. Recall that the terminal coalgebra $\nu H$ is considered as an algebra for $H$.
\begin{theorem}\label{thm:elgotbloom}
  Ih $H$ preserves countable coproduts and has a terminal coalgebra, then the following categories are isomorphic:
  \begin{enumerate}[(1)]
  \item the Eilenberg-Moore category $\C^\T$,
  \item the slice category $\nu H/\!\Alg H$
  \item the category of Bloom algebras for $H$, and
  \item the category of complete Elgot algebras for $H$.
  \end{enumerate}
\end{theorem}
\begin{proof}
  The isomorphism $(1) \cong (4)$ was proved in
  \cite[Theorem~5.8]{amv_elgot} for every iteratable endofunctor $H$. 

  The rest follows from various results in~\cite{ahm14}. In that paper we assumed that 
  $H$ is accessible and $\C$ is locally presentable. However, for our purposes we
  only apply those result of \emph{loc.~cit.}~ that do not depend on those assumptions, as we now explain. 
  First, the isomorphism $(2) \cong (3)$ was proved in \cite[Proposition~3.4]{ahm14} for
  every endofunctor $H$ having a terminal coalgebra $\nu H$. 

  The other results of \emph{loc.~cit.}~make use of
  coproducts in $\Alg H$. But since $H$ preserves countable
  coproducts, we know that the forgetful functor from $\Alg\, H$ to
  $\C$ creates countable coproducts. Hence, for example
  $TY = \coprod_{n < \omega} H^n Y + \nu H$ is a coproduct in
  $\Alg H$ of the free algebra $H^*Y = \coprod_{n< \omega} H^n Y$ on
  $Y$ and the algebra $(\nu H,
  t^{-1})$. By~\cite[Theorem~3.16]{ahm14}, $TY$ is then a free Bloom
  algebra on $Y$. That is, the forgetful functor $U_B$ of
  the category of Bloom algebra has the left adjoint $T(-)$.
  It is now easy to prove that $U_B$ is monadic,
  i.e., the isomorphism $(1) \cong (3)$ holds. The argument is given in
  the proof of~\cite[Theorem 4.15]{ahm14}; we repeat it here for the
  convenience of the reader (and to make clear that no extra
  assumptions are needed).

  Before we proceed let us recall~\cite[Lemma~3.7]{ahm14}: if
  $(A,a,\dagger)$ is a Bloom algebra and $h:(A,a)\rightarrow (B,b)$ is
  an algebra morphism, then there is a unique structure of a Bloom
  algebra on $(B,b)$ such that $h$ is a solution
  preserving algebra morphism. 

  We now prove that $U_B$ is monadic. By
  Beck's Theorem~\cite[4.4.4]{maclane}, it suffices to prove that
  $U_B$ creates coequalizers of $U_B$-split pairs. That means that
  given a parallel pair of solution preserving algebra morphisms
  \[
    f,g:(A,a,\dagger)\to (B,b,\ddagger)
  \]
  and given morphisms in $\C$ as follows
  \[
    \begin{array}{lp{6cm}}
      k: B\to C & \text{with $k\cdot f=k\cdot g$}, \\
      s:C\rightarrow B & \text{with $k\cdot s=id_C$, and} \\
      t:B\rightarrow A & \text{with $s\cdot k=f\cdot t$ and
                         $id_B=g\cdot t$,}
    \end{array}
  \]
  there exists a unique structure $(C,c,*)$ of a Bloom algebra such
  that $k$ is a solution preserving algebra morphism; moreover, $k$ is
  then a coequalizer in the category of Bloom algebras for $H$. Indeed, firstly,
  $C$ carries a unique
  structure of an $H$-algebra such that $k$ is an algebra morphism, namely:
  \[
    c = (HC \xrightarrow{Hs} HB \xrightarrow{b} B \xrightarrow{k} C)
  \]

  Secondly, by the above lemma there exists a unique structure $(C,c,*)$ of
  a Bloom algebra for which $k$ is a solution preserving algebra
  morphism. It only remains to verify that $k$ is a coequalizer in the
  category of Bloom algebras for $H$. To this end, let
  $h:(B,b,\ddagger)\to (D,d,+)$ be a solution preserving algebra
  morphism with $h\cdot f=h\cdot g$. There exists a unique algebra
  morphism $h':(C,c)\to (D,d)$ with $h=h'\cdot k$. In order to see
  that $h'$ preserves solutions (i.e., for every $e:X\rightarrow HX$
  we have $h'\cdot e^*=e^+$) we use that both $k$ and $h$ preserve
  solutions, and we calculate as follows:
  \[
    h'\cdot e^*  
    = 
    h'\cdot k\cdot e^\ddagger 
    =
    h\cdot e^\ddagger=e^+.\qedhere
  \]
\end{proof}
\fi

%
%
\section{Finitary $\Set$ Functors}

We have seen above that for every functor $H$ on a hyper-extensive
category preserving countable coproducts, the functors $H(-)+Y$ are
cia functors (i.e., every corecursive algebra is a cia). In
particular, if $\C$ is cartesian closed, then the functor
$X \mapsto W \times X + Y$ is a cia functor. For $\C = \Set$ and $H$
finitary we now prove the converse: if $H$ is a cia functor then it
has the form $X \mapsto W \times X + Y$ for some sets $W$ and $Y$.
\begin{assumption} \label{ass}
  Throughout this section $H$ denotes a standard, finitary set
  functor. 
\end{assumption}
Recall from~\cite{ap04} that $H$ is \emph{finitary} iff for every set
$X$ we have $HX = \bigcup HY$ where the union ranges over finite
subsets $Y \subseteq X$. An example of a finitary functor on $\Set$ is
the polynomial functor $H_\Sigma$, see Example~\ref{ex:corec}(3).

\emph{Standard} means that $H$ preserves 
\begin{enumerate}[(1)]
\item inclusions, i.e., $X \subseteq Y$ implies $HX \subseteq HY$ and
  the $H$-image of the inclusion map $X \subto Y$ is the inclusion
  map $HX \subto HY$, and
\item finite intersections.
\end{enumerate} 

Assuming that $H$ is standard is without loss of generality because
for every set functor $H$ there exist a standard set functor $H'$
naturally isomorphic to $H$ on the full subcategory of all nonempty
sets~\cite[Theorem~3.4.5]{at90}. (And the change of value at $\emptyset$
is irrelevant for us since corecursive algebras and cias,
respectively, for $H$ are in bijective correspondence with those for
$H'$).



\begin{definition}
\begin{enumerate}[(1)]
\item By a \emph {presentation} of $H$ is meant a finitary signature $\Sigma$ and natural
  epitransformation $\eps: H_\Sigma \rightarrow H$, i.e., every
  component $\eps_X$ is a surjective map.
\item An \emph {$\eps$-equation} is an expression
  $\sigma(x_1,\ldots x_n)=\tau(z_1, \ldots, z_m)$ where $\sigma$ is an
  $n$-ary operation symbol and $\tau$ an $m$-ary one such that
  $\eps_X$ merges the two elements of $H_\Sigma X$ where 
  $X=\{x_1, \ldots, x_n, z_1, \ldots, z_n\}$.
\end{enumerate}
\end{definition}
\begin{rem}
  \label{rem:eps}
  All $\eps$-equations form an equivalence relation. More precisely,
  for any set $X$ all $\eps$-equations with variables replaced by
  elements of $X$ form precisely the kernel equivalence of
  $\eps_X$. Moreover, the elements of $HX$ may be regarded as
  equivalence classes of the elements $\sigma(x_1, \ldots, x_n)$ of
  $H_\Sigma X$ modulo this equivalence.
\end{rem}
\begin{example}
  The finite power-set functor $\powf$ has a presentation with $\Sigma$
  having a single $n$-ary operation for every $n$, and $\eps$
  sending $\sigma (x_1, \dotsc ,x_n)$ to $ \{x_1, \dotsc ,x_n\}$.
\end{example}
The following lemma was proved in~\cite{at90}. We present a (short)
proof since we refer to it later.
\begin{lemma}\label{L:pres}
  Every finitary set functor has a presentation
  $\eps: H_\Sigma \to H$, and the category $\Alg H$ is isomorphic to
  the variety of all $\Sigma$-algebras satisfying all
  $\eps$-equations.
\end{lemma}
\begin{proof}
  Define a signature $\Sigma =(\Sigma_n )_{n \textless \omega}$ by
  $\Sigma_n= Hn$ where we regard $n$ as the finite ordinal
  $\{0, \dotsc ,n-1 \}$ for all $n$.
  By the Yoneda lemma we have a natural transformation
  $\eps_X : H_\Sigma X \rightarrow HX$ assigning to every
  $\sigma (x_1, \dotsc, x_n)$ represented as a function
  $x: n \to X$ the element $Hx(\sigma)$. Since
  $H$ is finitary, $\eps_X$ is surjective.

  Every $H$-algebra $a: HA \rightarrow A$ defines the corresponding
  $\Sigma$-algebra $a\cdot \eps_A : H_\Sigma A \rightarrow A$ which
  clearly satisfies all $\eps$-equations. This defines a full
  embedding of $\Alg H$ into $\Alg H_\Sigma$ (which is identity on
  morphisms). We now easily prove that every $\Sigma$-algebra
  satisfying all $\eps$-equations has the above form
  $(A, a\cdot\eps_A)$. Indeed, given $a^\Sigma: H_\Sigma A \to A$
  satisfying all $\eps$-equations, define $a: HA \to A$ by
  $a([\sigma(a_1, \ldots, a_n]) = a^\Sigma(\sigma(a_1, \ldots,
  a_n))$. Since we know from Remark~\ref{rem:eps} that $a^\Sigma$
  merges all pairs in the kernel of $\eps_A$, this is well-defined and
  we clearly have $a^\Sigma = a \cdot \eps_A$. Thus, our full embedding
  defines the desired isomorphism between $H$-algebras and
  $\Sigma$-algebras satisfing all $\eps$-equations.
\end{proof}

\begin{rem}\label{R:red}
\begin{enumerate}[(1)]
\item Denote by $C_1$ the constant functor with value $1=\{ c\}$, and
  by $C_{0.1}$ its subfunctor with $C_{0,1} \emptyset = \emptyset$ and
  $C_{0,1}X=1$ else. For every natural transformation
  $\alpha: C_{0,1} \rightarrow H$ there exists a unique extension to
  $\alpha': C_1 \rightarrow H$.

  Indeed, since $H$ is standard, it preserves the (empty) intersection
  of the coproduct injections $\inl,\inr: 1 \rightarrow 1+1$.  Since
  $H\inl(\alpha_1(c) )= \alpha_{1+1}(c)= H\inr(\alpha_1(c))$, there
  exists a unique element $t$ of $H \emptyset$ such that the inclusion
  map $v: \emptyset \rightarrow 1$ fulfils $\alpha_1(c)=Hv(t)$. We put
  $\alpha'_{\emptyset}(c)=t$.

\item All constants in our presentation of $H$ are
  \emph{explicit}. That means that whenever some $n$-ary symbol
  $\sigma$ has the property that some $\eps$-equation has the form
  $\sigma (x_1, \dotsc ,x_n)= \sigma (z_1, \dotsc ,z_n)$, where the
  variables $x_i$ are pairwise distinct and none of them equals some
  $z_j$, then there exists a constant symbol $\tau$ in $\Sigma$ for
  which we have the following $\eps$-equation:
  $\sigma (x_1, \dotsc ,x_n)= \tau$. Indeed, for every set
  $X \neq \emptyset$ we have an element
  \[
  \alpha _X= \eps _X(\sigma (a_1, \ldots ,a_n))  \in HX
  \]
  independent of the choice of $a_1, \ldots ,a_n$ in $X$. This defines
  a natural transformation $\alpha: C_{0,1} \to H$. Let
  $\alpha': C_1 \to H$ be its extension according to item~(1). The
  element $\alpha'_\emptyset (c)$ of $H\emptyset$ has, since $\eps$ is
  an epitransformation, the form $\eps _\emptyset (\tau)$ for some
  nullary symbol $\tau$. Then the desired $\eps$-equation holds
  because for $X = \{x_1, \ldots, x_n\}$ and the unique empty map
  $u: \emptyset \to X$ we have
  \[
    \eps_X(\sigma(x_1, \ldots, x_n)) = \alpha_X(c) = \alpha'_X(c) = Hu
    \cdot \alpha'_\emptyset(c) = Hu \cdot \eps_\emptyset(\tau) =
    \eps_X \cdot Hu(\tau) = \eps_X(\tau).
  \]
\end{enumerate}	
\end{rem}
\begin{definition}\label{D:red}
  A presentation $\eps: H_\Sigma \to H$ is \emph{reduced} provided
  that for every $\eps$-equation
  \[
    \sigma(x_1, \ldots, x_n) = \tau(z_1,\ldots, z_m)
  \]
  the following hold:
  \begin{enumerate}[(1)]
  \item if $x_1, \ldots, x_n$ are pairwise distinct, then they all lie in $\{z_1, \ldots, z_n\}$, and
  \item if, moreover, $z_1, \ldots, z_n$ are also pairwise distinct, then $\sigma =\tau$.
  \end{enumerate}
\end{definition}
\begin{proposition}\label{P:red}
  Every finitary set functor has a reduced presentation.
\end{proposition}
\iffull
\begin{proof}
  (a)~Assume that the above condition (1) holds. Then we can restrict
  $\eps$ so that also (2) becomes true.  Indeed, denote by $\sim$ the
  following equivalence on $\Sigma$: $\sigma \sim \tau$ iff there
  exists an $\eps$-equation
  $\sigma(x_1, \ldots, x_n) = \tau(z_1,\ldots, z_m)$ with pairwise
  distinct variables on both sides. Condition (1) implies that $n=m$
  and there exists a permutation $(i_1, \ldots,i_n)$ with
  $x_1=z_{i_1}, \ldots, x_n=z_{i_n}$. This implies that the image of the
  summand $\{\sigma\} \times X^n$ under $\eps_X$ is equal to the image
  of $\{\tau \} \times X^m$. Consequently, if $\Sigma'$ is a choice
  class of $\sim$, then the restriction $\eps'$ of $\eps$ to
  $H_\Sigma'$, as a subfunctor of $H_\Sigma$, is still an
  epi-transformation. And the presentation $\eps'$ fulfils (1) and
  (2) in Definition~\ref{D:red}.

  (b)~It remains to prove that every presentation $\eps$ can be modified
  to one satisfying (1) in Definition~\ref{D:red}.  Let $\sigma$ be an
  $n$-ary symbol of $\Sigma$. For $i=1, \ldots, n$ we say that the
  coordinate $i$ is \emph{inessential} for $\sigma$ if we have an
  $\eps$-equation of the following form:
  \[
    \sigma (x_1, \ldots ,x_n) = \sigma (x_1, \ldots , x_{i-1}, z, x_{i+1}, \ldots , x_n)
  \]
  all of whose $n+1$ variables are pairwise distinct. The remaining
  coordinates will be called \emph{essential}.  Without loss of
  generality we can assume that the essential coordinates are
  precisely $1, \ldots, n'$ for some $n' \leq n$. From Remark \ref{R:red}(3)
  it follows easily that the following is also an $\eps$-equation:
  \[
    \sigma (x_1, \ldots ,x_n) = \sigma (x_1, \ldots , x_{n'}, z, \ldots , z).
  \]
  Form the signature $\Sigma'$ with the same symbols as $\Sigma$ but
  with arities $n'$ in lieu of $n$.  We define a presentation
  $\eps': H_{\Sigma'} \to H$ as follows: for each nonempty set $X$ it
  sends every element $\sigma(x_1, \ldots ,x_{n'})$ to
  $\eps_X (\sigma (x_1, \ldots ,x_{n'},z, \ldots, z))$, where $z$ is
  arbitrary. And to define $\eps_{\emptyset}$, use
  Remark~\ref{R:red}(2): whenever a symbol $\sigma$ has no essential
  coordinate (and hence $\sigma$ becomes a constant symbol in
  $\Sigma'$), there exists a constant symbol $\tau$ in $\Sigma$ and an
  $\eps$-equation $\sigma(x_1, \ldots, x_n) = \tau$. Define
  $\eps'_{\emptyset}(\sigma) = \eps_{\emptyset}(\tau)$. This
  presentation $\eps'$ clearly satisfies both conditions of
  Definition~\ref{D:red}.
\end{proof}
\fi
\begin{notation}\label{N:not}
  From now on we assume that a reduced presentation of $H$ is given. 
  
  Recall the notation $TY$, $FY$ and $CY$ from Examples~\ref{ex:cia}
  and Notation~\ref{not:FC}. All these objects exist since $H$ is
  finitary (and therefore so are all $H(-) + Y$). The corresponding
  notation for $H_\Sigma$ is $T_\Sigma Y$, $F_\Sigma Y$ and
  $C_\Sigma Y$. The monad units of $T$ and $C$ are denoted by $\eta $
  and $\eta^C$, respectively.
  
  As mentioned above, $T_\Sigma Y$ can be described as the algebra of
  all $\Sigma$-trees over $Y$. And $C_\Sigma Y$ and $F_\Sigma Y$ are
  its subalgebras on all trees with finitely many leaves labeled in
  $Y$, or all finite trees, respectively.
 
  Since $TY$ is a corecursive algebra, there exists a unique
  homomorphism of $H$-algebras
  \[
    m_Y :CY \to TY
  \]
  with $m_Y \cdot \eta^C_Y = \eta_Y$. The corresponding
  $H_\Sigma$-algebra morphism is denoted by
  \[
    m^\Sigma_Y :C_\Sigma Y \to T_\Sigma Y.
  \]
\end{notation} 
\begin{rem}\label{R:AM}
  In ~\cite{am06} we described $FY$ and $TY$ as the
  following quotient of the $\Sigma$-algebras $F_\Sigma Y$ and
  $T_\Sigma Y$, respectively. Recall from Lemma~\ref{L:pres} that
  every $H$-algebra $a: HA \to A$ may be regarded as the
  $H_\Sigma$-algebra with structure
  $a \cdot \eps_A: H_\Sigma A \to A$.
  \begin{enumerate}[(1)]
  \item $FY = F_\Sigma Y / \mathord{\sim_Y}$, where $\sim_Y$ is the
    congruence of finite application of $\eps$-equations. That is, the
    smallest congruence with $\sigma(x_1, \ldots, x_n) \sim_Y \tau(z_1,
    \ldots, z_m)$ for every $\eps$-equation 
    \[
      \sigma(x_1, \ldots, x_n) = \tau(z_1,\ldots, z_m)
    \]
    over $Y$. The universal map $\eta^F_Y: Y
    \to FY$ is the composition of the one of $F_\Sigma Y$ with the
    canonical quotient map $F_\Sigma Y \epito F_\Sigma Y/\mathord{\sim_Y}$. 
  \item $TY = T_\Sigma Y/\mathord{\sim^*_Y}$, where $\sim^*_Y$ is the
    congruence of (possibly infinitely many) applications of
    $\eps$-equations. The universal map is $\wh\eta_Y
    = \wh\eps_Y \cdot \eta^\Sigma_Y$, where $\eta^\Sigma_Y: Y \to
    T_\Sigma Y$ is the universal map of the free cia for $H_\Sigma$ on
    $Y$ and $\wh\eps_Y: T_\Sigma Y \epito T_\Sigma Y/\mathord{\sim^*_Y}$
    is the canonical quotient map. 
  \end{enumerate}
\end{rem}
The definition of a possibly \emph{infinite application of
  $\eps$-equations} is based on the concept of \emph{cutting} a
$\Sigma$-tree at level $k$: the resulting finite $\Sigma$-tree
$\partial_k t$ is obtained from $t$ by deleting all nodes of depth
larger than $k$ and relabeling all nodes at level $k$ by a symbol
$\bot \not\in Y$. Then we define, for $\Sigma$-trees $t$ and $s$ in
$T_\Sigma Y$,
\[
  t \sim^*_Y s
  \qquad
  \text{iff}
  \qquad
  \partial_k t \sim_{Y \cup \{\bot\}} \partial_k s
  \quad
  \text{for every $k < \omega$}.
\] 

Not surprisingly, $CY$ can be described analogously:
\begin{proposition}\label{P:C}
The free corecursive $H$-algebra $CY$ is the quotient of the $\Sigma$-algebra $C_\Sigma Y$ modulo the application of $\eps$-equations:
$CY=C_\Sigma Y/\mathord{\sim^*_Y}$.
\end{proposition}
\begin{proof}
  This is based on the following description of $CY$ presented
  in~\cite{ahm14}: denote by $\oplus$ the binary coproduct of
  $H$-algebras in $\Alg H$. By Lemma~\ref{L:pres}, this is,
  equivalently, the coproduct in the variety of all $\Sigma$-algebras
  satisfying all $\eps$-equations. Then we have
  \[
    CY=\nu H \oplus FY.
  \]
  Analogously, if $\boxplus$ denotes the binary coproduct of
  $\Sigma$-algebras, we of course have
  \[
    C_\Sigma Y=\nu H_\Sigma \boxplus F_{\Sigma} Y.
  \]
  For arbitrary $H$-algebras $A$ and $B$ we know that $A \oplus B$ is
  the quotient of $A \boxplus B$ modulo the application of $\eps$-equations.
  Moreover, we have $T=T_\Sigma /\mathord{\sim^*}$ and
  $FY=F_\Sigma Y/\mathord{\sim}$. It follows immediately that
  $T \oplus FY=(T_\Sigma \boxplus F_\Sigma Y)/\mathord{\sim^*}$, as claimed.
\end{proof}
\begin{lemma}\label{L:square}
  Suppose that $CY$ is a cia for $H$. For every equation morphism $e:
  X \to H_\Sigma X + Y$ with the unique solution $\ssol e: X
  \to T_\Sigma Y$ we can form an equation morphism
  \[
    \ol e = (X \xrightarrow{e} H_\Sigma X + Y
    \xrightarrow{\eps_X+\eta^C_Y} HX + CY).
  \]
  \iffull
  Then the square below commutes:
  \begin{equation}
    \label{eq:sq}
    \xymatrix{
      X \ar[r]^-{\sol{\ol e}} \ar[d]_{\ssol e} & CY \ar[d]^{m_Y}\\
      T_\Sigma Y \ar[r]_-{\hat\eps_Y} & TY
    }
  \end{equation}
  \else
  Then we have 
  $
  (X \xrightarrow{\sol{\ol e}} CY \xrightarrow{m_Y} TY) 
  =
  (X \xrightarrow{\ssol e} T_\Sigma Y \xrightarrow{\hat\eps_Y} TY)
  $.
  \fi
\end{lemma}
\iffull
\begin{proof}
  Put 
  \[
    \wt e = (X \xrightarrow{e} H_\sigma X + Y \xrightarrow{\eps_X +
      \eta_Y} HX + TY).
  \]
  We prove that both sides of the square~\eqref{eq:sq} are solutions of $\wt e$ in
  the cia $TY$ for $H$. 

  (1)~That $\wh\eps_Y \cdot \ssol{\ol e}$ solves $\wt e$ is due to the
  following diagram:
  \[
    \xymatrix@C+2pc{
      X
      \ar[d]_e 
      \ar[r]^-{\ssol e}
      &
      T_\Sigma Y 
      \ar[r]^-{\hat\eps_Y}
      & 
      TY
      \\
      H_\Sigma X + Y
      \ar[r]_-{H_\Sigma \ssol e + Y}
      \ar[d]_{\eps_X + \eta_Y}
      &
      H_\Sigma T_\Sigma Y + Y 
      \ar[u]_{[\tau^\Sigma_Y, \eta^\Sigma_Y]}
      \ar[d]^{\eps_{TY} + \eta_Y}
      \ar[r]_-{H_\Sigma \hat\eps_Y + Y}
      &
      H_\Sigma TY + Y
      \ar[u]^{[\tau_Y \cdot \eps_{TY}, \eta_Y]}
      \ar[d]_{\eps_{TY} + \eta_Y}
      \\
      HX + TY 
      \ar[r]_-{H\ssol e + Y}
      \ar@{<-} `l[u] `[uu]^{\tilde e} [uu]
      &
      HT_\Sigma Y + Y
      \ar[r]_-{H\hat\eps_Y + TY}
      &
      HTY + TY
      \ar `r[u] `[uu]_{[\tau_Y, TY]} [uu]
    }
  \]
  The left-hand part commutes by the definition of $\wt e$, and the
  right-hand part does trivially. The upper left-hand square commutes
  by the definition of $\ssol e$. For the lower two ones consider the
  coproduct components separately: the left-hand one commutes since
  $\eps$ is natural, and the right-hand one trivially does. And for
  the remaining upper right-hand part one considers the coproduct
  components separately once more: the right-hand one states that
  $\wh\eps_Y \cdot \eta^\Sigma_X = \eta_Y$, and for the left-hand one
  we use that $TY$ considered as an $H_\Sigma$-algebra (with the
  structure $\tau_Y \cdot \eps_{TY}$) is a quotient of the free
  $H_\Sigma$-algebra $(T_\Sigma Y, \tau^\Sigma_Y)$ via the quotient
  algebra morphism $\wh\eps_Y$ as explained in Remark~\ref{R:AM}(2).

  (2)~That $m_Y \cdot \sol{\ol e}$ solves $\wt e$ is due to the
  following diagram:
  \[
    \xymatrix@C+2pc{
      X
      \ar[rr]^{\sol{\ol e}}
      \ar[d]_e
      && 
      CY
      \ar[r]^-{m_Y}
      & 
      TY 
      \\
      H_\Sigma X + Y
      \ar[rd]^-{\eps_X + \eta^C_Y}
      \ar[dd]_{\eps_X + \eta_Y}
      \\
      & 
      HX + CY
      \ar[r]_-{H\sol{\ol e} + CY}
      \ar[ld]_-*+{\labelstyle HX + m_Y}
      &
      HCY + CY
      \ar[d]^{HCY + m_Y}
      \ar[uu]_{[\psi_Y, CY]}
      \\
      HX + TY
      \ar@{<-} `l[u] `[uuu]^{\tilde e} [uuu]
      \ar[rr]_-{H\sol{\ol e} + TY}
      &&
      HCY + TY
      \ar[r]_-{Hm_Y + TY}
      &
      HTY + TY \ar[uuu]_{[\tau_Y, TY]}
    }
  \]
  The left-hand part commutes by the definition of $\wt e$, and the
  upper left inner part commutes since $\sol{\ol e}$ is a solution of
  $\ol e$. For the triangle on the left consider the coproduct
  components separately: the right-hand one commutes since
  $m_Y \cdot \eta^C_Y = \eta_Y$ (see Notation~\ref{N:not}), and the
  left-hand component trivially commutes; the middle lower part
  obviously commutes.  Finally, for the right-hand part consider the
  coproduct components separately ones more: the left-hand component
  commutes since $m_Y$ is an $H$-algebra morphism from $(CY, \psi_Y)$
  to $(TY, \tau_Y)$, and the right-hand component trivially commutes.
\end{proof}
\fi
\begin{theorem}\label{T:finitary}
  For a finitary set functor $H$ the following conditions are
  equivalent:
  \begin{enumerate}
  \item $H$ is a cia functor,
  \item $H = H_0(-) + Y$ where $H_0$ preserves countable coproducts
    and $Y$ is a set, and
  \item $H = W \times (-) + Y$ for some sets $W$ and $Y$. 
  \end{enumerate}
\end{theorem}
\begin{proof}
  (2) $\Rightarrow$ (3). Since $H$ is finitary, so is $H_0$, by
  the description of finitarity following Assumptions~\ref{ass}. Therefore, $H_0$ preserves all
  coproducts. Trnkov\'a proved~\cite[Theorem~IX.8]{t71}, that every
  coproduct-preserving set functor preserves colimits, thus it is a
  left adjoint.  It is well known that the only right adjoint set
  functors $R$ are the representable ones: for given $L \dashv R$, put
  $W=L1$, then the elements $1 \to RY$ bijectively correspond to the
  maps $W \to Y$, thus, $R$ is naturally isomorphic to
  $\Set(W,-)$. Consequently, $H_0$ is left adjoint to $\Set(W,-)$, hence
  it is naturaly isomorphic to $W \times (-)$.

  (3) $\Rightarrow$ (1). This follows from Corollary~\ref{C:strong}.

  (1) $\Rightarrow$ (2). Let $\eps: H_\Sigma \to H$ be a reduced
  presentation. 
  
  (a)~We prove below that all arities in $\Sigma$ are 1 or 0.  Let $W$
  be the set of all unary symbols and $Y$ that of all constants.  Then
  $H_\Sigma X=W \times X + Y$. Furthermore, we show that $\eps$ is a
  natural isomorphism. Indeed, each $\eps_X$ is, besides being
  surjective, also injective: it cannot merge distinct elements
  $(w,x)$ and $(w',x')$ of $W \times X$ because this would yield an
  $\eps$-equation $w(x)=w'(x')$. Since the presentation is reduced,
  this implies $w = w'$ and $x=x'$. Analogously for all other pairs of
  elements of $H_\Sigma X$.

  \medskip 
  (b)~Assume that some symbol $\alpha$ of $\Sigma$ has arity
  at least $2$. Then we derive a contradiction to $H$ being a cia
  functor. Given a $\Sigma$-tree $t$ we call a node $r$ \emph{pure}
  if the trees $t_1, \ldots, t_n$ rooted at the children of $r$ are
  pairwise distinct\iffull\/:
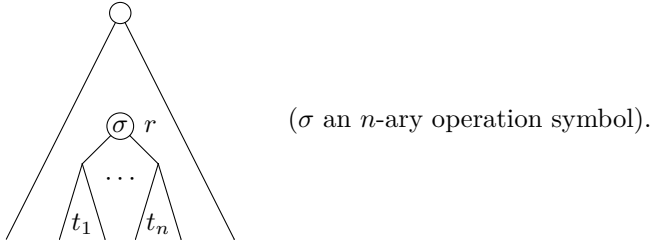
  \[
    \vcenter{
      \xy
      \POS (000,000) *+[o]=<8pt>[F-]{ } = "o",
           (-15,-30) = "l",
           (015,-30) = "r",
           (000,-15) *+[o]=<10pt>[F-]{\sigma} = "s",
           (004,-15) *{r},
           (-05,-20) = "lt",
           (005,-20) = "rt",
           (-08,-30) = "llt",
           (-02,-30) = "rlt",
           (-05,-28) *{t_1},
           (002,-30) = "lrt",
           (008,-30) = "rrt",
           (005,-28) *{t_n},
           (000,-22) *{\cdots}
      \ar@{-} "o";"l"
      \ar@{-} "o";"r"
      \ar@{-} "s";"lt"
      \ar@{-} "s";"rt"
      \ar@{-} "lt";"llt"
      \ar@{-} "lt";"rlt"
      \ar@{-} "rt";"lrt"
      \ar@{-} "rt";"rrt"
      \endxy
    }
    \qquad
    \text{($\sigma$ an $n$-ary operation symbol).} 
  \]\else\/. \fi
  Observe that an $\eps$-equation applicable to a pure node $r$ must
  have the form
  \iffull\[\else\/$\fi \sigma(x_1, \ldots, x_n) = \tau(y_1, \ldots,
    y_m) \iffull\]\else\/$ \fi for some $\tau \in \Sigma_m$, where
  $x_1, \ldots, x_n$ are pairwise distinct.

  Consider the following equation morphism $e: X \to H_\Sigma X + Y$ with
  $X = \{x_1, \ldots, x_n\}$ and $Y = \{y_2, \ldots, y_n\}$:
  \iffull
  \[
    e(x_1) = \alpha(x_1, y_2, y_n)
    \qquad\text{and}\qquad
    e(x_i) = y_i\quad
    \text{for $i = 2, \ldots, n$.}
  \]
  \else
  $e(x_1) = \alpha(x_1, y_2, y_n)$ and $e(x_i) = y_i$, for $i = 2, \ldots, n$. \fi
  Then the unique solution  $\ssol e:X \to T_\Sigma Y$ assigns to $x_1$
  the $\Sigma$-tree below:
  \[
    \ssol e (x_1) = 
    \vcenter
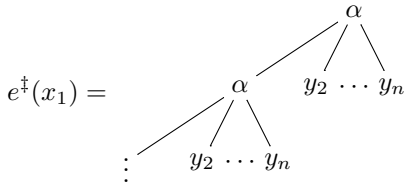
{
    \xy
      \POS (000,000) *+{\alpha} = "a1",
           (-15,-10) *+{\alpha} = "a2",
           (-05,-10) *+{y_2} = "12",
           (005,-10) *+{y_n} = "1n",
           (000,-10) *{\cdots},
           (-20,-20) *+{y_2} = "22",
           (-10,-20) *+{y_n} = "2n",
           (-15,-20) *{\cdots},
           (-30,-20) *+{\vdots} = "a3"
       \ar@{-} "a1";"a2"
       \ar@{-} "a1";"12"
       \ar@{-} "a1";"1n"
       \ar@{-} "a2";"a3"
       \ar@{-} "a2";"22"
       \ar@{-} "a2";"2n"
    \endxy
    }
  \]
  Next consider the equation morphism
  \iffull\[\else\/$\fi
    \ol e = (X \xrightarrow{e} H_\Sigma X + Y \xrightarrow{\eps_X +
      \eta^C_Y} HX + CY).
  \iffull\]\else\/$ \fi
  Since $CY$ is a cia, this has a unique solution $\sol e: X \to CY$. It
  assigns to $x_1$ an element of $CY$ which by Proposition~\ref{P:C}
  has the form
  \iffull
  \[
    \sol{\ol e}(x_1) = \ol\eps_Y(s)
    \qquad
    \text{for some $s \in C_\Sigma Y$},
  \]
  \else
  $\sol{\ol e}(x_1) = \ol\eps_Y(s)$ for some $s \in C_\Sigma Y$, 
  \fi
  where $\ol\eps_Y: C_\Sigma Y \epito C_\Sigma Y/\mathord{\sim^*} \cong
  CY$ denotes the canonical quotient map. From Lemma~\ref{L:square} we
  know that
  \[
    \wh\eps_Y(t) 
    =
    \wh\eps_Y \cdot \ssol e(x_1) 
    =
    m_Y \cdot \sol{\ol e}(x_1) 
    = 
    m_Y \cdot \ol\eps_Y(s).
  \]
  Therefore, we obtain $t \sim^*_Y s$.

  We derive the desired contradiction by proving that every tree
  obtained from $t$ by a finite application of $\eps$-equations has a
  leaf labeled by $y_2$ at every positive level. From this we conclude 
  immediately that the same holds for all trees obtained by an
  infinite application of $\eps$-equations from $t$. However, $t
  \sim^*_Y s$ where $s$ has only finitely many leaves labeled by
  $y_2$.

  \medskip (b1)~Assume that a single $\eps$-equation is applied to $t$ and let $t'$ be
  the resulting tree. Let $r$ be the node of $t$ at which the
  application takes place. Then $r$ is not a leaf labeled in $Y$; for
  recall that all $\eps$-equations  have operation
  symbols on both sides, thus, they are not applicable to leaves
  labeled in $Y$. Therefore, $r$ is a pure node labeled by $\alpha$. The
  $\eps$-equation in question thus has the form
  \iffull\[\else\/$\fi
    \alpha(u_1,\ldots, u_n) = \tau(z_1, \ldots, z_m)
  \iffull\]\else\/$ \fi
  for some $\tau \in \Sigma_m$ and with the $u_i$ pairwise distinct.

  If $r$ has depth $k$, then the tree $t'$ has label $y_2$ at all
  levels $1, \ldots, k$, since those leaves of $t$ are
  unchanged. Furthermore, we have $u_2 = z_p$ for some
  $p= 1, \ldots, m$ since $\eps$ is a reduced presentation. Therefore,
  $y_2$ occurs at level $k+1$ since the $p$-th child of $r$ in $t'$ is
  a leaf labeled by $y_2$. For the levels greater than $k+1$ we use
  that $u_1 = z_q$ holds for some $q= 1, \ldots, m$, again because
  $\eps$ is a reduced presentation. Since the first subtree of $r$ in
  $t$ is $t$ itself, it follows that the $q$-th child of $r$ in $t'$
  is $t$ itself. Thus, a label $y_2$ of depth $n$ in $t$ yields a
  label $y_2$ of depth $k+1+n$ of $t'$.

  \medskip (b2)~Assume that two $\eps$-equations are applied to $t$. The resulting
  tree $t''$ can be obtained from $t'$ in~(b1) by a single application
  of an $\eps$-equation. Let $r'$ be the node of $t'$ at which the
  application takes place. We can assume $r \neq r'$ (for if $r = r'$
  we can obtain $t''$ from $t$ by a single application on an
  $\eps$-equation; this follows from Remark~\ref{rem:eps}). If
  $r'$ does not lie in the subtree of $t'$ with root $r$, then
  $r'$ is a pure node labeled by $\alpha$ and we argue as in~(b1).

  Suppose therefore that $r'$ lies in the subtree rooted at $r$. If
  this is the $q$-th subtree from~(b1) above (the one with $u_1 =
  z_q$), then we also argue as in~(b1) using that the $q$-th subtree
  is $t$ itself. Otherwise, if $r'$ lies in any other subtree of $r$,
  then the labels $y_2$ of the $q$-th subtree are unchanged. 

  The remaining cases of three and more applications of
  $\eps$-equations are completely analogous. This yields the desired
  contradiction: if $t \sim^* \ol t$, then $\ol t$ has label $y_2$ at
  every level $1, 2, 3, \ldots$, thus $t\sim^*_Y s$ cannot be true. 
\end{proof}


\section{Conclusions and Open Problems}

For endofunctors $H$ preserving countable coproducts and having a
terminal coalgebra we have described the free corecursive algebra on
an object $Y$ as $\nu H +\coprod_{n < \omega} H^n Y$. In addition, we
have shown that $H$ is a cia functor, i.e., every corecursive algebra
for $H$ is a cia. For this we assumed that the base category
has well-behaved countable coproducts, i.e., the category is
hyper-extensive. It is an open problem whether our results hold in
more general categories, e.g., in all extensive locally presentable
ones.

For accessible functors $H$ on locally presentable categories, the
free corecursive algebra on $Y$ was described in previous
work~\cite{ahm14} as the coproduct of $FY$ (the free algebra on $Y$)
and $\nu H$ (considered as an algebra) in the category $\Alg H$. If
$H$ preserves countable coproducts, this is quite similar to the above
desciption of the free cia, since coproducts of algebras are then
formed on the level of the underlying category and therefore
$FY= \coprod_{n < \omega} H^n Y$. But the proof techniques are
completely different, and a common generalization of the two results
is open.

We have also characterized all cia functors among finitary set
functors: they are precisely the functors $X \mapsto W \times X + Y$
for some sets $W$ and $Y$. In Example~\ref{E:U} we have seen that
the same result does not hold for all, not necessarily finitary, set
functors. But that example required an assumption about set theory. It
is an open problem whether that assumption was really necessary.
 
Our results can be stated in terms of corecursive monads~\cite{ahm14}
and completely iterative ones~\cite{aamv} as follows: a functor $H$
having a terminal coalgebra $\nu H$ and preserving countable
coproducts has a free corecursive monad of the form
$\coprod_{n < \omega} H^n(-) + \nu H$, and this is also the free
completely iterative monad on $H$.

%
%
\bibliographystyle{plainurl}
\bibliography{refs}

\end{document}